\documentclass{llncs}

\usepackage{booktabs}
\usepackage{pgf}
\usepackage{tikz}
\usepackage{etoolbox}
\usetikzlibrary{arrows,automata}
\usepackage{amsmath}
\usepackage{array}
\usepackage{stmaryrd}
\usepackage{todonotes}
\usepackage[utf8]{inputenc}
\usepackage{url}
\usepackage{alltt}
 
\usepackage{amssymb}
\usepackage[english]{babel}
\usepackage{proof}
\usepackage{mathpartir}
\usepackage{framed}
\usepackage{esvect}

\usepackage{tikz}
\usetikzlibrary{trees}
\usetikzlibrary{arrows}
\usetikzlibrary{decorations.pathmorphing}
\usetikzlibrary{shapes.multipart}
\usetikzlibrary{shapes.geometric}
\usetikzlibrary{calc}
\usetikzlibrary{positioning} 
\usetikzlibrary{fit}
\usetikzlibrary{backgrounds}

\pgfkeys{/pgf/rectangle split parts=10}
\newcommand{\alltrue}[0]{\bot}

\newcommand{\secref}[1]{Section~\ref{#1}}

\newcommand{\figref}[1]{Figure~\ref{#1}}

\newcommand{\theoremref}[1]{Theorem~\ref{#1}}

\newcommand{\exampleref}[1]{Example~\ref{#1}}

\newcommand{\Grs}{LDGRS}
\newcommand{\lhs}{L}
\newcommand{\str}{\texttt{s}}

\newcommand{\verule}{\rule{1pt}{10pt}\kern 0.2em}

\newcommand{\rin}[0]{L_{in}}
\newcommand{\rout}[0]{L_{out}}
\newcommand{\rlin}[0]{L_{l\_in}}
\newcommand{\rlout}[0]{L_{l\_out}}
\newcommand{\rll}[0]{L_{l\_loop}}
\newcommand{\ein}[0]{E'_{in}}
\newcommand{\eout}[0]{E'_{out}}
\newcommand{\elin}[0]{E'_{l\_in}}
\newcommand{\elout}[0]{E'_{l\_out}}
\newcommand{\elol}[0]{E'_{l\_loop}}
\newcommand{\fin}[0]{in}
\newcommand{\fout}[0]{out}
\newcommand{\flin}[0]{l\_in}
\newcommand{\flout}[0]{l\_out}
\newcommand{\fll}[0]{l\_loop}
\newcommand{\phiin}[0]{\phi_{in}}
\newcommand{\phiout}[0]{\phi_{out}}
\newcommand{\philin}[0]{\phi_{l\_in}}
\newcommand{\philout}[0]{\phi_{l\_out}}
\newcommand{\phill}[0]{\phi_{l\_loop}}

\newcommand{\cin}[0]{c_{in}}
\newcommand{\cout}[0]{c_{out}}
\newcommand{\clin}[0]{c_{l\_in}}
\newcommand{\clout}[0]{c_{l\_out}}
\newcommand{\cll}[0]{c_{l\_loop}}

\newcommand{\Co}[0]{\mathcal{C}}
\newcommand{\Ro}[0]{\mathcal{R}}
\newcommand{\cel}[0]{\mathcal{C}_0}
\newcommand{\rel}[0]{\mathcal{R}_0}
\newcommand{\phin}[0]{\Phi_N}
\newcommand{\phir}[0]{\Phi_E}
\newcommand{\phing}[0]{\Phi_N^G}
\newcommand{\phirg}[0]{\Phi_E^G}
\newcommand{\phinp}[0]{\Phi_N^{G'}}
\newcommand{\phirp}[0]{\Phi_E^{G'}}

\newlength{\MMtextNodeWidth}

\newtoggle{SorL}

\toggletrue{SorL}

\title{On the Verification of Logically Decorated Graph
  Transformations}

\author{
Jon Ha\"el Brenas\inst{1}
\and
    Rachid Echahed\inst{2}
\and
   Martin Strecker\inst{3}
}

\institute{
  UTHSC - ORNL,
  Memphis, Tennessee, USA\\
  \email{jhael@uthsc.edu}
\and
   CNRS and University Grenoble-Alpes,
   Grenoble, France\\
   \email{rachid.echahed@imag.fr}
\and
   Université de Toulouse, IRIT Institute, Toulouse, France,\\
   \email{martin.strecker@irit.fr}
 }

\authorrunning{J.H. Brenas and R. Echahed}

\titlerunning{Verification of Logically Decorated Graph
  Transformations}

\begin{document}

\maketitle

\begin{abstract}
We address the problem of reasoning on graph transformations featuring
actions such as \emph{addition} and \emph{deletion} of nodes and
edges, node \emph{merging} and \emph{cloning}, node or edge
\emph{labelling} and edge \emph{redirection}. First, we introduce the
considered graph rewrite systems which are parameterized by a given
logic $\mathcal{L}$. Formulas of $\mathcal{L}$ are used to label graph
nodes and edges.  In a second step, we tackle the problem of formal
verification of the considered rewrite systems by using a Hoare-like
weakest precondition calculus. It acts on triples of the form
$\{\texttt{Pre}\}(\texttt{R},\texttt{strategy}) \{\texttt{Post}\}$
where \texttt{Pre} and \texttt{Post} are conditions specified in the
given logic $\mathcal{L}$, \texttt{R} is a graph rewrite system and
\texttt{strategy} is an expression stating how rules in \texttt{R} are
to be performed. We prove that the calculus we introduce is sound.
Moreover, we show how the proposed framework can be instantiated
successfully with different logics.  We investigate first-order logic
and several of its decidable fragments with a particular focus on
different dialects of description logic (DL). We also show, by using
bisimulation relations, that some DL fragments cannot be used due to
their lack of expressive power.


\end{abstract}

\section{Introduction}\label{sec:intro}
Graphs, as well as their transformations, play a central role in
modeling  data in various areas such as chemistry, civil
engineering or computer science. In many such applications, it
may be desirable to be able to prove that the transformations are
correct, i.e., from any graph (or state) satisfying a given set of conditions,
only graphs satisfying another set of conditions can be obtained. 

In this paper, we address the problem of correctness of programs
defined as graph rewrite rules. The correctness properties are stated
as logical formulas obtained using a Hoare-like calculus. The
considered graph structures are attributed with logical formulas which
label both nodes and edges. Definitions of the structures as well as
their transformation are provided in a generic framework parameterized
by a given logic $\mathcal{L}$. Rewrites rules follow an algorithmic
approach where the left-hand sides are attributed graphs and the
right-hand sides are sequences of elementary actions
\cite{Echahed08b}. Among the considered actions, we quote node and
edge \emph{addition} or \emph{deletion}, node and edge
\emph{labelling} and edge \emph{redirection}, in addition to node
\emph{merging} and \emph{cloning}.  To our knowledge, the present work
is the first to consider the verification of graph transformations
including the last two actions, namely node merging and node cloning.
We propose a sound Hoare calculus for the considered specifications
defined as triples of the form
$\{\texttt{Pre}\}(\texttt{R},\texttt{strategy}) \{\texttt{Post}\}$
where \texttt{Pre} and \texttt{Post} are conditions specified in a
given logic $\mathcal{L}$, \texttt{R} is a graph rewrite system and
\texttt{strategy} is an expression stating how rules in \texttt{R} are
to be performed. Different instances of the logic $\mathcal{L}$ are
provided in this paper in order to illustrate the effectiveness of the
proposed method.

The correctness of graph transformations has attracted some attention
in recent years. One prominent approach is model checking such as
the Groove tool
\cite{ghamarian_mol_rensink_zambon_zimakova_2012}. The idea is to
carry out a symbolic exploration of the state space, starting from a
given graph, in order to find out whether certain invariants are
maintained or certain states are
reachable.  The Viatra tool has similar model checking capabilities
\cite{journals/sosym/Varro04} and in addition allows the verification
of elaborate well-formedness constraints imposed on graphs
\cite{semerath_barta_szatmari_horvath_varro_2015}. Well-formedness is
within the realm of our approach (and amounts to checking the
consistency of a formula), but is not the primary goal of this paper
which is on the dynamics of graphs.
The Alloy analyser \cite{jackson_software_abstractions_2011} uses
bounded model checking for exploring relational designs and
transformations (see for example \cite{Baresi2006} for an application
to graph transformations). Counter-examples are presented in graphical
form.  The aforementioned techniques are sometimes combined with
powerful SAT- or SMT-solvers, but do not carry out a complete
deductive verification, though.

Hoare-like calculi for the verification of graph transformations have
already been proposed with different logics to express the pre- and
post-conditions.  Among the most prominent approaches figure nested
conditions \cite{HP09,PP10} that are explicitly created to describe
graph properties. The considered graph rewrite transformations are
based on the double pushout approach with linear spans which forbid
actions such as node merging and node cloning.

Other logics might be good candidates to express graph properties
which go beyond first-order definable properties such as monadic
second-order logic \cite{C90,PP14} or the dynamic logic defined in
\cite{BEH10} which allows one to express both rich graph properties as
well as the graph transformations at the same time.  These approaches
are undecidable in general and thus either cannot be used to prove
correctness of graph transformations in an automated way or only work
on limited classes of graphs.

Starting from the other side of the logical spectrum, one could
consider the use of decidable logics such as fragments of Description
Logics to specify graph properties \cite{ACOS14,BES16}.
Decidable fragments of first-order logics such as two-variable logic
with counting and logics with exists-forall-prefix, among others, can
be of practical use as well in the verification of graph transformation
\cite{DBLP:journals/jar/PiskacMB10,boerger_graedel_gurevich_decision_problem,graedel_otto_rosen_lics97}.
  
 The paper is organized as follows. Formal preliminary 
definitions of the considered graph structures and the elementary transformation
actions are introduced in the next section. 
 In \secref{sec:GRS}, we
define the investigated class of graph rewrite systems and the used
notion of rewrite strategies.  The proposed
Hoare-calculus for the verification of the correctness of graph transformations
is presented \secref{sec:verification}. In \secref{sec:logic}, some
logics that can be used for the considered verification
problems are presented. 
We also point out some fragments of Description Logic whose expressive
power is not sufficient enough to be useful in reasoning on graph
dynamics. Concluding remarks are given in \secref{sec:conclusion}. The
missing proofs can be found in the appendix.


\section{Preliminaries}\label{sec:GGT}
We start by introducing the notion of \emph{logically decorated
  graphs}.  Nodes and edges of such graph structures are labeled by
logic formulas. The definition below is parameterized by a given logic
$\mathcal{L}$ seen as a set of formulas. Section~\ref{sec:logic}
provides some examples of possible candidates for such a logic
$\mathcal{L}$.

\begin{definition}[Logically Decorated Graph]
\label{def:graphs}
Let $\mathcal{L}$ be a logic (set of formulas). A \emph{graph alphabet} is a pair ($\Co$,
$\Ro$) of sets of elements of $\mathcal{L}$, that is
$\Co \subseteq \mathcal{L}$ and $\Ro \subseteq
\mathcal{L}$. $\Co$ is the set
of \emph{node formulas} 
or \emph{concepts} and $\Ro$ is
the set of \emph{edge formulas} 
or \emph{roles}\footnote{The names \emph{concept} and \emph{role} are
  borrowed from  Description Logics' vocabulary~\cite{DLHandbook}.}. Subsets of $\Co$ and $\Ro$, respectively named $\cel$
and $\rel$, contain basic (propositional) concepts and roles
respectively.  A \emph{logically decorated graph} $G$ over a graph
alphabet ($\Co$, $\Ro$) is a tuple ($N$, $E$, $\phin$, $\phir$, $s$,
$t$) where $N$ is a set of \emph{nodes}, $E$ is a set of \emph{edges},
$\phin$ is the \emph{node labeling} function,
$\phin: N \rightarrow \mathcal{P}(\Co)$, $\phir$ is the \emph{edge
  labeling} 
function, $\phir: E \rightarrow \Ro$,
$s$ is the \emph{source function} $s: E \rightarrow N$ and $t$ is the
\emph{target function} $t: E \rightarrow N$.
\end{definition}

Transformation of logically decorated graphs, considered in the next
section, will be defined following an algorithmic approach based on
the notion of \emph{elementary actions} as introduced below. These
actions constitute a set of elementary graph transformations such as
the addition/deletion of nodes, concepts or edges ; redirection
of edges ; merge or clone of nodes. Formal definitions of the
considered elementary actions are given  in \figref{semaact}.

\begin{definition}[Elementary action, action]
An \emph{elementary action}, say $a$, may be of the following forms:
\begin{itemize} 
\item a \emph{node addition} $add_N(i)$  (resp. \emph{node
    deletion} $del_N(i)$) where $i$ is a new node (resp. an existing node). It creates the node $i$. $i$ has no incoming nor
  outgoing edge and it is not labeled with any concept (resp. it deletes $i$ and all its incoming and outgoing edges). 

\item a \emph{concept addition} $add_C(i,c)$ (resp. \emph{concept
    deletion} $del_C(i,c)$) where $i$ is a node
  and $c$ is a basic concept (a proposition name) in $\cel$. It adds the label $c$ to (resp. removes the label $c$ from) the labeling of node $i$.
\item an \emph{edge addition} $add_E(e,i,j,r)$ (resp. \emph{edge
    deletion} $del_E(e,i,j,r)$) where $e$ is an edge, $i$ and $j$ are nodes
  and $r$ is a basic role (edge label) in $\rel$. It adds the edge $e$
  with label $r$ between nodes $i$ and $j$ (resp. removes the edge $e$). When the edge that is affected is clear from the context, we will usually simply write $add_E(i,j,r)$ (resp. $del_E(i,j,r)$).
\item a \emph{global edge redirection} $i \gg j$ where $i$ and $j$ are
  nodes. It redirects all incoming edges of $i$ towards $j$.
\item a \emph{merge action} $mrg(i, j)$ where $i$ and $j$ are nodes. This action merges the two nodes. It yields a new graph in which the first node $i$ is labeled with the union of the labels of $i$ and $j$ and  such that all incoming or outgoing edges of any of the two nodes are gathered.
\item a \emph{clone action} $cl(i,j,\rin,\rout,\rlin,\rlout,\rll)$
  where $i$ and $j$ are nodes and $\rin$, $\rout$, $\rlin$, $\rlout$
  and $\rll$ are sets of basic roles. It clones a node $i$ by creating a new node $j$ and connect $j$ to the rest of a host graph according to different information given in the parameters $\rin,\rout,\rlin,\rlout,\rll$.
\end{itemize}

The result of performing an elementary action $a$ on a graph $G =
(N^G,E^G,\\C^G,R^G,\phing,\phirg,s^G,t^G)$, written
$G[\alpha]$, produces the graph $G' =
(N^{G'},E^{G'},C^{G'},\\R^{G'},\phinp,\phirp,s^{G'},t^{G'})$ as
defined  in \figref{semaact}.
An \emph{action}, say $\alpha$, is a
sequence of elementary actions of the form $\alpha = a_1 ; a_2 ; \ldots ;
a_n$. The result of performing $\alpha$ on
a graph $G$ is written $G[\alpha]$. $ G[a ; \alpha] = (G[a])[\alpha]$
and $G[\epsilon] = G$ where $\epsilon$ is the empty sequence.
\end{definition}

The elementary action $cl(i,j,\rin,\rout,\rlin,\rlout,\rll)$ might be
not easy to grasp at first sight. It thus deserves some
explanations. Let node $j$ be a clone of node $i$. What would be the
incident edges of the clone $j$? answering this question is not
straightforward. There are indeed different possibilities to connect
$j$ to the neighborhood of $i$.  \figref{fig:automata} illustrates
such a problem : there are indeed different possibilities to connect
node $q'_1$, a clone of node $q_1$, to the other nodes. In order to
provide flexible clone action, the user may tune the way the edges
connecting a clone are
treated through the five parameters
$\rin,\rout,\rlin,\rlout,\rll$. All these parameters are subsets of
the set of basic roles $\Ro_0$ and are explained informally below :

\begin{itemize}
\item $\rin$ indicates that every incoming edge $e$ of $i$ which is
  not a loop and whose label
  is in $\rin$ is cloned as a new edge $e'$
  such that $s(e') = s(e)$ and $t(e') = j$. 
\item $\rout$ indicates that every outgoing edge $e$ from $i$ which
  not a loop and whose label
  is in $\rout$  is cloned as a new edge $e'$
  such that $s(e') = j$ and $t(e') = t(e)$. 
\item $\rlin$ indicates that every self-loop $e$ over $i$ whose label
  is in $\rlin$  is cloned as a new edge $e'$
  such that $s(e') = i$ and $t(e') = j$. (e.g., see the blue arrow in
  \figref{fig:automata})

\item $\rlout$ indicates that every self-loop $e$ over $i$ whose label
  is in $\rlout$  is cloned as a new edge $e'$
  such that $s(e') = j$ and $t(e') = i$. (e.g., see the red arrow in
  \figref{fig:automata})

\item $\rll$ indicates that every self-loop $e$ over $i$ whose label
  is in $\rll$  is cloned as a new edge $e'$
  which is a self-loop over $j$, i.e, $s(e') = j$ and $t(e') =
  j$. (e.g., see the selfloop over node $q_1'$ in
  \figref{fig:automata})
\end{itemize}

Additionally, the semantics of the cloning actions as defined in
\figref{semaact} use several sets of edges, representing the edges
that are created depending on how they should be connected to $i$ and
$j$, permitting to associate to each new edge the old one of which it
is a copy. The sets $\ein, \eout, \elin, \elout$ and $\elol$, used in
\figref{semaact}, are pairwise disjoint sets of new (fresh) edges, and
the functions $\fin, \fout, \flin, \flout$ and $\fll$ are bijections
defined such that:
\begin{enumerate}

\item $\ein$ is in bijection through function \emph{$\fin$} with the set
  $\{e \in E^G |\; t^G(e) = i \wedge s^G(e) \neq i \wedge \phirg(e) \in
  \rin\}$,

\item $\eout$ is in bijection through function \emph{$\fout$} with the
  set
  $\{e \in E^G |\; s^G(e) = i \wedge t^G(e) \neq i \wedge \phirg(e) \in
  \rout\}$,

\item $\elin$ is in bijection through function \emph{$\flin$} with the
  set
  $\{e \in E^G |\; s^G(e) = t^G(e) = i \wedge \phirg(e) \in \rlin\}$,

\item $\elout$ is in bijection through function \emph{$\flout$} with
  the set
  $\{e \in E^G |\; s^G(e) = t^G(e) = i \wedge \phirg(e) \in \rlout\}$,

\item $\elol$ is in bijection through function \emph{$\fll$} with the set $\{e \in E^G |\; s^G(e) = t^G(e) = i \wedge \phirg(e) \in \rll\}$).
\end{enumerate}

Informally, the set $\ein$ contains a copy of every incoming edge $e$
of node $i$, i.e. such that $t^G(e) = i$, which is not a self-loop,
i.e. $s^G(e) \neq i$, and having a label in $\rin$, i.e.
$\phirg(e) \in \rin$.  $\rin$ is thus used to select
which incoming edges are cloned.  The other sets
($\eout, \elin, \elout$ and $\elol$) are defined similarly.

\begin{figure}
\scalebox{.99}{
\begin{tabular}{l|l}
\bf{If $\alpha = add_C(i,c)$ then:}& \bf{If $\alpha =
del_C(i,c)$ then:}\\
$N^{G'} = N^G$,$E^{G'} = E^G$, & $N^{G'} = N^G$,$E^{G'} = E^G$, \\
$\phinp(n) =\left\{
	\begin{array}{ll}
		\phing(n) \cup \{c\}  & \mbox{if } n = i \\
		\phing(n) & \mbox{if } n \neq i
	\end{array}
\right.$ & $\phinp(n) =\left\{
	\begin{array}{ll}
		\phing(n) \backslash \{c\}  & \mbox{if } n = i \\
		\phing(n) & \mbox{if } n \neq i
	\end{array}
\right.$\\
$\phirp = \phirg$, $s^{G'} = s^G$, $t^{G'} = t^G$ & $\phirp = \phirg$, $s^{G'} = s^G$, $t^{G'} = t^G$\\
\bf{If $\alpha = add_E(e,i,j,r)$ then:}&\bf{If $\alpha = del_E(e,i,j,r)$ then:}\\
$N^{G'} = N^G$,  $\phinp = \phing$ & $N^{G'} = N^G$,  $\phinp = \phing$\\
$E^{G'} = E^G \cup \{e\}$ &$E^{G'} = E^G \backslash \{e\}$ \\
$\phirp(e') = \left\{
	\begin{array}{ll}
		 r & \mbox{if } e' = e \\
		\phirg(e') & \mbox{if } e' \neq e
	\end{array}
\right.$ & \text{$\phirp$ is the restriction of $\phirg$ to $E^{G'}$}\\
$s^{G'}(e') = s^G (e')$ if $e'\not=e, s^{G'}(e) = i$ & \text{$s^{G'}$ is the restriction of $s^G$ to $E^{G'}$}\\
$t^{G'} (e') = t^G(e')$ if $e'\not=e, t^{G'}(e) = j$ &     \text{$t^{G'}$ is the restriction of $t^G$ to $E^{G'}$}\\
\bf{If $\alpha = add_N(i)$ then:} & \bf{If $\alpha = cl(i,j,\rin,\rout,\rlin,\rlout,\rll)$ then:} \\
$N^{G'} = N^G \cup \{i\}$ where $i$ is a new node & $N^{G'} = N^G \cup \{j\}$ \\
$E^{G'} = E^G$, $\phirp = \phirg$,$s^{G'} = s^G$, $t^{G'} = t^G$ & $E^{G'} = E^G \cup \ein \cup \eout \cup \elin \cup \elout \cup \elol$ \\
$\phinp(n) = \left\{
	\begin{array}{ll}
		\emptyset & \mbox{if } n = i \\
		\phing(n) & \mbox{if } n \neq i
	\end{array}
\right.$ & $\phinp(n) = \left\{
	\begin{array}{ll}
		\phing(i) \cap \mathcal{C}_0  & \mbox{if } n = j\\
		\phing(n) & \mbox{otherwise}
	\end{array}
\right.$ \\
$\left. \begin{array}{l}
    \text{\bf{If $\alpha = del_N(i)$ then:}}\\
    N^{G'} = N^G \backslash \{i\} \\
    E^{G'} = E^G \backslash \{e | s^G(e) = i \vee t^G(e) = i\}\\
    \text{$\phinp$ is the restriction of $\phing$ to $N^{G'}$}\\
    \text{$\phirp$ is the restriction of $\phirg$ to $E^{G'}$}\\
    \text{$s^{G'}$ is the restriction of $s^G$ to $E^{G'}$}\\
    \end{array}
    \right.$ & $\phirp(e) = \left\{
	\begin{array}{ll}
	  \phirg(\fin(e))  & \mbox{if } e \in \ein\\
          \phirg(\fout(e))  & \mbox{if } e \in \eout\\
          \phirg(\flin(e))  & \mbox{if } e \in \elin\\
          \phirg(\flout(e))  & \mbox{if } e \in \elout\\
          \phirg(\fll(e))  & \mbox{if } e \in \elol\\          
          \phirg(e) & \mbox{otherwise}
	\end{array}
\right.$ \\
$\left. \begin{array}{l}
    \text{$t^{G'}$ is the restriction of $t^G$ to $E^{G'}$}\\
    \text{\bf{If $\alpha = i \gg j$ then:}}\\
    \text{$N^{G'} = N^G$, $E^{G'} = E^G$}\\
    \text{$\phinp = \phing$, $\phirp = \phirg$,$s^{G'} = s^G$}\\
    t^{G'}(e) = \left\{
        \begin{array}{ll}
            j & \mbox{if } t^G(e) = i \\
            t^G(e) & \mbox{if } t^G(e) \neq i
        \end{array}
    \right.
    \end{array} \right.$ & 
$s^{G'}(e) = \left\{
	\begin{array}{ll}
          s^G(\fin(e))  & \mbox{if } e \in \ein\\
          j  & \mbox{if } e \in \eout\\
          i  & \mbox{if } e \in \elin\\
          j  & \mbox{if } e \in \elout\\
          j  & \mbox{if } e \in \elol\\          
          s^G(e) & \mbox{otherwise}
        \end{array}
\right.$ \\
$\left. \begin{array}{l}
    \text{\bf{If $\alpha = mrg(i,j)$ then:}}\\
    \text{$N^{G'} = N^G \backslash \{j\}, E^{G'} = E^G, \phirp(e) = \phirg(e)$}\\
    \phinp(n) = \left\{
        \begin{array}{ll}
            \phing(i) \cup \phing(j)  & \mbox{if } n = i\\
            \phing(n) & \mbox{otherwise}
        \end{array}\right.\\
    s^{G'}(e) = \left\{
        \begin{array}{ll}
        i  & \mbox{if } s^G(e) = j\\
        s^G(e) & \mbox{otherwise}
        \end{array} \right.
\end{array} \right.$ & $t^{G'}(e) = \left\{
	\begin{array}{ll}
	  j  & \mbox{if } e \in \ein\\
          t^G(\fout(e))  & \mbox{if } e \in \eout\\
          j  & \mbox{if } e \in \elin\\
          i  & \mbox{if } e \in \elout\\
          j  & \mbox{if } e \in \elol\\
          t^G(e) & \mbox{otherwise}
	\end{array}
    \right.$\\
    $t^{G'}(e) = \left\{
        \begin{array}{ll}
            i  & \mbox{if } t^G(e) = j\\
            t^G(e) & \mbox{otherwise}
        \end{array}
    \right.$
\end{tabular}
}
\caption{$G' = G[\alpha]$, summary of the effects of the elementary actions:
  $add_N(i)$, $del_N(i)$, $add_C(i,c)$, $del_C(i,c)$, $add_E(e,i,j,r)$, $del_E(e)$, $i \gg j$, $mrg(i,j)$ and $cl(i,j,\rin,\rout,\rlin,\rlout,\rll)$. $\Co$ and $\Ro$ are never modified.}\label{semaact}
\end{figure}

Notice that a node and its clone have the same basic (propositional)
labels (see \figref{semaact}). It is possible to define alternate
versions of cloning regarding node labels (clone none of them, add a
parameter stating which ones to clone, etc.). The results presented in
this paper can be extended to these other definitions of cloning in a
straightforward manner.

\begin{example}\label{ex:automata}
  Let $\mathcal{A} = \{Q,\Sigma,\delta,q_0,F\}$ be the automaton of
  \figref{fig:automata} A. Performing the action
  $cl(q_1,q'_1,\Sigma,\Sigma, X,Y,Z)$ gives the
  automaton presented in B where the blue - plain - (resp. red -
  dashed -, purple - dotted) transition exists iff $X$ (resp.
  $Y$, $Z$) contains the label $\{a\}$. 
\end{example}
\begin{figure}
\def\smallscale{0.6}
\begin{center}
\resizebox{5cm}{!}{
\begin{tikzpicture}[->,>=stealth',shorten >=1pt,auto,node distance=2.8cm,
                    semithick,group/.style ={fill=gray!20, node distance=20mm},thickline/.style ={draw, thick, -latex'}]
  \tikzstyle{Cstate}=[circle,fill=white,draw=black,text=black, minimum
  size = 1cm]
  \tikzstyle{NCstate}=[circle,fill=none,draw=none,text=black, minimum
  size = 1cm]
  \tikzstyle{CurrState}=[circle,fill=white,draw=black,text=black, minimum
  size = 1cm]

  \node[Cstate] at (0,0) (Q0) {\huge $q_0$};
  \node[Cstate] at (0,2) (Q1) {\huge $q_1$};
  \node[NCstate] at (-.8,2) (O0) {};
  \node[Cstate] at (0,4) (Q2)       {\huge $q_2$};

  \node[Cstate] at (4,0) (Q'0) {\huge $q_0$};
  \node[Cstate] at (2.5,2) (Q'1) {\huge $q_1$};
  \node[Cstate] at (5.5,2) (R1) {\huge $q'_1$};
  \node[NCstate] at (1.7,2) (O1) {};
  \node[NCstate] at (6.5,2) (O2) {};
  \node[Cstate] at (4,4) (Q'2) {\huge $q_2$};

\begin{pgfonlayer}{background}
\node [label=below:{A},group, fit=(Q0) (O0) (Q2)] (b) {};
\node [label=below:{B},group, fit=(Q'0) (O1) (O2) (Q'2)] (b') {};
\end{pgfonlayer}

\path (Q0) edge              node {b} (Q1)
(Q1) edge          node {b} (Q2)
(Q1) edge [loop left] node {a} (Q1)
(Q'0) edge              node {b} (Q'1)
(Q'1) edge          node {b} (Q'2)
(Q'1) edge [loop left] node {a} (Q'1)
(Q'0) edge              node {b} (R1)
(R1) edge  node {b} (Q'2);

\path[blue] (Q'1) edge[bend right] node {a} (R1);

\path[red,dashed] (R1) edge[bend right] node {a} (Q'1);

\path[purple,dotted,thickline] (R1) edge[loop right] node {a} (R1);

\end{tikzpicture}
}
\end{center}
\caption{A) An automaton and B) the possible results of cloning node
  $q_1$ as node $q'_1$}
\label{fig:automata}
\end{figure}
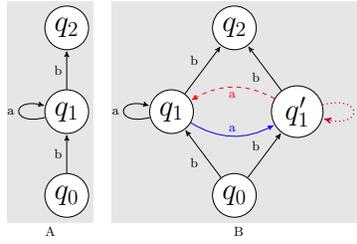

\begin{example}
\label{ex:weord}

Let $(V, \geq)$ be a (weak) ordering on a finite set $V$.  The
relation $\geq$ is reflexive, asymmetric and transitive. Let us
consider the representation of this ordering as a graph having $V$ as
a set of nodes and there is an edge $e$ labeled with $\geq$ between
every two nodes (elements) $i$ and $j$ of $V$ iff $i \geq j$. Then, creating an
 clone of a node does not make sense in general if one
wishes to keep the correspondence between the edges of the graph and
the property of ordering. However cloning can be  used to create an element that is the next
smaller element (if performing
$cl(n,n',\{\leq\},\{\leq\},\{\leq\},\emptyset,\{\leq\})$) or the next greater
element (if performing
$cl(n,n',\{\leq\},\{\leq\},\emptyset, \{\leq\},\emptyset)$).
\end{example}

Readers familiar with algebraic approaches to graph transformation may
recognize the cloning flexibility provided by the recent PBPO
(pullback-pushout) approach of \cite{CorradiniDEPR17}. The parameters
of the clone action reflect somehow the typing morphisms of
\cite{CorradiniDEPR17}.  Cloning a node according to the approach of
Sesquipushout \cite{CorradiniHHK06} could be easily simulated by
instantiating all the parameters by the full set of basic roles
$cl(i,j,\rel,\rel,\rel,\rel,\rel)$.

Another action which may affect several edges in a row is the merge
action.  \figref{fig:mrgex} illustrates an example of node merging. To
be more precise, node $j$ of the left graph is merged with node $i$.
Notice that, except for the name of the resulting node ($i$ in this
case), $mrg(i,j)$ and $mrg(j,i)$ are the same.  After the action is
performed, the edges between nodes $i$, $l$ and $k$ already present
before the merge action remain unchanged and a new edge is added
between $i$ and $k$ inherited from the link between $j$ and $k$. A
loop over $i$ is also added representing the edge between $i$ and $j$
in the initial graph.

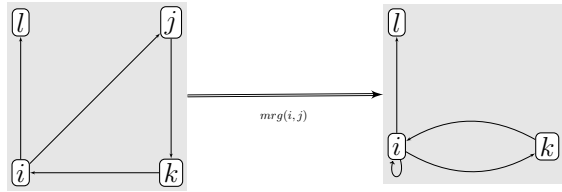
\begin{figure}
\begin{center}
\scalebox{.5}{
\begin{tikzpicture}
[auto,
blockbl/.style ={rectangle, draw=black, fill=white, thick,
  text centered, rounded corners,
  minimum height=1em
},
blockb/.style ={rectangle, draw=blue, thick, fill=blue!20,
  text width=2em, text centered, rounded corners,
  minimum height=1em
},
blockr/.style ={rectangle, draw=red, thick, fill=red!20,
  text centered, rounded corners,
  minimum height=1em
},
blockg/.style ={rectangle, draw=green, thick, fill=green!20,
  text width=2em, text centered, rounded corners,
  minimum height=1em
},
blockw/.style ={rectangle, draw=blue, thick, fill=white!20,
  text centered, rounded corners,
  minimum height=1em
},
blockwh/.style ={rectangle, draw=blue, thick, fill=white!20,
  text centered, rounded corners, text=white,
  minimum height=3em
},
blocke/.style ={rectangle, draw=none, thick, fill=none,
  text centered, rounded corners,
  minimum height=1em
},
group/.style ={fill=gray!20, node distance=10mm},
ggroup/.style ={fill=red!5, node distance=10mm},
igroup/.style ={fill=red!15, node distance=10mm},
line/.style ={draw=black, -latex'},
lineb/.style ={draw=blue, -latex'},
liner/.style ={draw=red, -latex'},
lineg/.style ={draw=green, -latex'},
thickline/.style ={draw, thick, double, -latex', }
]

\node (a0) at (0,0) [blockbl] {\huge $i$};
\node (b0) at (0,4) [blockbl] {\huge $l$};
\node (c0) at (4,0) [blockbl] {\huge $k$};
\node (d0) at (4,4) [blockbl] {\huge $j$};

\node (a1) at (10,0.7) [blockbl] {\huge $i$};
\node (b1) at (10,4) [blockbl] {\huge $l$};
\node (c1) at (14,0.7) [blockbl] {\huge $k$};

\node (A1) at (7,1.5) [blocke] {$mrg(i,j)$};
\node (oo) at (12,0) [blocke] {};
\begin{pgfonlayer}{background}
\node [group, fit=(a0) (b0) (c0) (d0)] (G0) {};
\node [group, fit=(a1) (b1) (c1) (oo)] (Act0) {};
\end{pgfonlayer}

\begin{scope}
\draw [thickline] (G0) -- (Act0);
\draw [line] (a0) -- (b0);
\draw [line] (c0) -- (a0);
\draw [line] (a0) -- (d0);
\draw [line] (d0) -- (c0);
\draw [line] (a1) -- (b1) ;
  \path
    (a1) edge [line, loop below] node {} (a1)
           edge [line, bend right] node {} (c1)
    (c1) edge [line, bend right] node {} (a1);

\end{scope}
\end{tikzpicture}
}
\end{center}
\caption{Example of application of the elementary action Merge}\label{fig:mrgex}
\end{figure}

%
%

\section{Graph Rewriting Systems and Strategies}\label{sec:GRS}
In this section, we introduce the notion of \emph{logically decorated
grapg  rewriting systems}, or \Grs. These are extensions of the graph rewriting
systems defined in \cite{Echahed08b} where graphs are attributed with formulas from a given logic. The left-hand sides of the rules are thus logically
decorated graphs whereas the right-hand sides are defined as sequences
of elementary actions.

\begin{definition}[Rule, \Grs]
A \emph{rule} $\rho$ is a pair ($\lhs$, $\alpha$) where $\lhs$, called
the left-hand side, is a logically decorated graph
and $\alpha$, called the right-hand side,
is an action. Rules are usually written $\lhs \rightarrow \alpha$. A
logically decorated graph rewriting system, \Grs, is a set of rules.
\end{definition}

Let us point out that the left-hand side of a rule is an attributed
graph, that is it can contain nodes labeled with formulas. This is not
insignificant as these formulas  express additional conditions to be
satisfied during the matching process, e.g. reachability (graph
accessibility) condition, constraints on the number of neighbors (counting
quantifiers), etc. depending on the underlying logic. In the sequel,
we will use the symbol $\models$ to indicate the satisfiability
relation between items of a graph (nodes or edges) and logical formulas. 

\begin{definition}[Match]\label{match}
A \emph{match} $h$  between a left-hand side $\lhs$ and a graph $G$ is
a pair of functions $h = (h^N, h^E)$, with $h^N : N^{\lhs}
\rightarrow N^G$ and  $h^E : E^{\lhs} \rightarrow E^G$ such that:\\
\begin{tabular}{llll}
1.& $\forall n \in N^{\lhs}, \forall c \in \Phi_N^{\lhs}(n), h^N(n)
  \models c$\hspace*{.5cm} &
  2.& $\forall e \in E^{\lhs}, \Phi_E^G(h^E(e)) = \Phi_E^{\lhs}(e) $\\
3.& $\forall e \in E^{\lhs}, s^G(h^E(e)) =   h^N(s^{\lhs}(e))$ & 4.& $\forall
                                                          e \in E^{\lhs},
                                                          t^G(h^E(e))
                                                          = h^N(t^{\lhs}(e))$

\end{tabular}
\end{definition}

The third and the fourth conditions are classical and say that the
source and target functions and the match have to agree. The first
condition says that for every node $n$ of the left-hand side, the node
to which it is associated, $h(n)$, in $G$ has to satisfy every concept
in $\Phi_N^{\lhs}(n)$. This condition clearly expresses additional
negative and positive conditions which are added to the ``structural''
pattern matching. The second one ensures that the match respects edge
labeling.

\begin{definition}[Rule application]
A graph $G$ rewrites to graph $G'$ using a rule  $\rho = (\lhs, \alpha)$
iff there exists a match $h$ from $\lhs$ to $G$. $G'$ is obtained from
$G$ by performing actions in $h(\alpha)$\footnote{$h(\alpha)$ is
  obtained from $\alpha$ by replacing every node name, $n$, of $\lhs$
  by $h(n)$.}.
Formally, $G' = G[h(\alpha)]$. We write $G \to_{\rho} G'$. 
\end{definition}

Confluence of graph rewrite systems is not easy to establish. For
instance, orthogonal graph rewrite systems are not always
confluent, see e.g. \cite{Echahed08b}. We use the notion of rewrite
strategies to control the use of possible rules. Informally, a strategy
specifies the application order of different rules. 
It does not point to where the matches are to be
found nor does it ensure unique normal forms.

\begin{definition}[Strategy]
Given a graph rewriting system $\mathcal{R}$, a \emph{strategy} is a word of the
following language defined by $s$, where $\rho$ is any rule in $\mathcal{R}$:\\
\begin{tabular}{lll|ll|ll}
$s :=$ & $\epsilon$ & (Empty strategy) \hspace*{1cm}& $\rho$ & (Rule) \hspace*{1cm}& $s \oplus s$ & (Choice)
  \\
& $s; s$ & (Composition) & $s^*$ & (Closure) & $\rho?$ & (Rule trial)
 \\
& $\rho!$ & (Mandatory Rule)\\
\end{tabular}\\

\end{definition}

Informally, the strategy $"s_1 ; s_2"$ means that strategy $s_1$ should be
applied first, followed by the application of strategy $s_2$. On the other hand, $s_1 \oplus s_2$ means that either the strategy $s_1$ or the strategy $s_2$ is applied. The strategy $\rho^*$ means that rule $\rho$
is applied as many times as possible. Notice that the closure is the standard
``while'' construct: if the strategy we use is $s^*$, the strategy $s$ is used as long as
it is possible and not an undefined number of times.
The strategies $\rho$, $\rho?$ and $\rho!$ try to apply the rule
$\rho$. They behave in the same way when the rule $\rho$ matches the
host graph. However, when rule $\rho$ does not match the host graph,
the strategy written $\rho$ ends the rewriting process successfully.
The strategy $\rho?$, called \emph{Rule
Trial}, simply skips the application of the rule and the rewriting process
proceeds to the following strategy. The strategy $\rho!$, named 
\emph{Mandatory Rule}, stops and the rewriting process fails.

We write $G \Rightarrow_{\str} G'$ to denote that graph $G'$ is
obtained from $G$ by applying the strategy $\str$.  In
\figref{fig:semrules}, we provide the rules that specify how
strategies are used to rewrite a graph. For that we use the following
atomic formula $\bold{App}$($\str$) such that for all graphs $G$, $G
\models \bold{App}(\str$) iff the strategy $\str$ can perform at
least one step over $G$. This atomic formula is defined below.

\begin{tabular}{llllll}\label{Apstrat}
$\hspace*{-1em}\bullet\hspace*{1em} G \models \bold{App}(\rho)$ & iff there exists a match $h$ from the
  left-hand side of $\rho$ to $G$\\
$\hspace*{-1em}\bullet\hspace*{1em} G \models \bold{App}(\rho!)$ &  iff there exists a match $h$ from the
  left-hand side of $\rho$ to $G$\\
$\hspace*{-1em}\bullet\hspace*{1em} G \models \bold{App}(\epsilon)$ & $\hspace*{1em}\bullet\hspace*{1em}
                                                G \models \bold{App}(s_0 \oplus
                                                s_1)$ iff $G \models \bold{App}(s_0)$ or $G \models \bold{App}(s_1)$\\
$\hspace*{-1em}\bullet\hspace*{1em} G \models \bold{App}(s_0^*)$ & 
$\hspace*{1em}\bullet\hspace*{1em} G \models \bold{App}(s_0;s_1)$ iff
                                                        $G \models \bold{App}(s_0)$\\
$\hspace*{-1em}\bullet\hspace*{1em} G \models \bold{App}(\rho?)$
\end{tabular}

Notice that $G \models \bold{App}(s)$ does not mean that the whole
strategy can be applied on $G$, but just its first step can be
applied.  Indeed, let us assume the strategy $s = (s_0;s_1)$
where $s_0$ can be applied but may yield a state where $s_1$
cannot. In this case, the strategy $s$ can be applied on $G$ ($G
\models \bold{App}(s)$) but the execution
may stop after performing one step of $s_0$.
 
\begin{figure}
\begin{framed}
\begin{scriptsize}
\begin{tabbing}
1\=11111111111111111111111111111111\=1111111111111111111111111111111111111111111111111111\=111111111111\kill\\
\>\infer[(\text{Empty rule})]{G \Rightarrow_{\epsilon} G}{%
 } \\ \ \\
 \infer[(\text{Strategy composition})]{G \Rightarrow_{s_0;s_1} G'}{%
   G \Rightarrow_{s_0} G'' \quad G'' \Rightarrow_{s_1} G'}\\ \ \\
\>\infer[(\text{Choice left})]{G \Rightarrow_{s_0 \oplus s_1} G'}{%
   G \Rightarrow_{s_0} G'} \>
 \infer[(\text{Choice right})]{G \Rightarrow_{s_0 \oplus s_1} G'}{%
  G \Rightarrow_{s_1} G'}\\ \ \\
\>\infer[(\text{Closure false})]{G \Rightarrow_{s^*} G}{%
   G \not\models \bold{App}(s)} \> \infer[(\text{Closure true})]{G \Rightarrow_{s^*} G'}{%
  G \Rightarrow_{s} G'' \quad G'' \Rightarrow_{s^*} G' \quad G \models \bold{App}(s)}\\ \ \\
\>\infer[(\text{Rule False})]{G \Rightarrow_{\rho} \alltrue}{%
  G \not\models \bold{App}(\rho)} \>
\infer[(\text{Rule True})]{G \Rightarrow_{\rho} G'}{%
  G \models \bold{App}(\rho) \quad  G \to_{\rho} G'} \\ \ \\  
\>\infer[(\text{Mandatory Rule False})]{G \not\Rightarrow_{\rho!} }{%
   G \not\models \bold{App}(\rho)} \>  \hspace*{1cm}\infer[(\text{Mandatory Rule True})]{G \Rightarrow_{\rho!} G'}{%
  G \models \bold{App}(\rho) \quad  G \to_{\rho} G'}
\\ \ \\
\> 
 \infer[(\text{Rule Trial False})]{G \Rightarrow_{\rho?} G}{%
   G \not\models \bold{App}(\rho)} \>\hspace*{1cm}\infer[(\text{Rule Trial True})]{G \Rightarrow_{\rho?} G'}{%
  G \models \bold{App}(\rho) \quad  G \to_{\rho} G'}

\end{tabbing}
\end{scriptsize}
\end{framed}
\caption{Strategy application rules}\label{fig:semrules}
\end{figure}

The three strategies using rules (i.e. $\rho$, $\rho!$ and $\rho?$)
behave the same way when $G \models App(\rho)$ holds, ss shown in
\figref{fig:semrules}, but they do differ when 
$G \not\models App(\rho)$. In such a case, $\rho$ can yield any
graph, denoted by $\alltrue$, (i.e. the process stops without an
error), $\rho!$ stops the rewriting  process with failure and $\rho?$
ignores the rule application and  moves to the
next step of the execution of the strategy.

\begin{example}\label{ex:servernet}
Let us assume that we are managing a set of servers. Clients can connect to proxy servers that are themselves connected to mail servers, print servers, web servers, etc. We use graph transformations to generate new proxy servers to avoid over- or under-use of proxy servers. The rules that are used are shown in \figref{fig:servernet}. They use the description logic $\mathcal{ALCQUOI}$ introduced in \secref{sec:logic}. For this example, actions that affect an edge, e.g. $add_E(e,i,j,r)$, identify an edge from its extremities, e.g. we will write $add_E(i,j,r)$ instead.

Both rules select a $Client$ that $Request$ed a connection to a
$Proxy$. If the proxy has less than $N$ currently established
client-to-proxy ($C2P$) connections (rule $\rho_0$), the label
$Request$ is removed and the label $C2P$ is added to the edge between
the $Client$ and the $Proxy$. If the $Proxy$ already has more than $N$
client-to-proxy connections (rule $\rho_1$), the $Proxy$ is
cloned. All its incoming edges, except for those labeled with
$Request$ or $C2P$, are cloned as well as all outgoing
edges. Self-loops are not cloned. The label $Request$ is then dropped from the edge between the $Client$ and the original $Proxy$ and the edge from the $Client$ to the new $Proxy$ is labeled with $C2P$. 

The application condition for the first rule, $\bold{App}(\rho_0)$, is
$\exists U. (i \wedge Client \wedge \exists Request. (j \wedge Proxy
\wedge (<\; N\; C2P^-\; \top)))$\footnote{This formula is actually not
  as expressive as $\bold{App}(\rho_0)$. This problem is discussed
  more in \secref{sec:logic}.}. This condition can be understood as
"there exists a node named $i$ labeled with $Client$ that is the
source of an edge labeled with $Request$ whose target is a node named
$j$ labeled with $Proxy$ and such that there are strictly less than
$N$ different nodes connected through $C2P$ to$j$".
The used strategy is $s = \rho_0 \oplus \rho_1$ i.e. either $\rho_0$ or $\rho_1$ is applied but not both.
\end{example}

\begin{figure}
  \begin{center}
\scalebox{.8}{
    \begin{tikzpicture}
[auto,
blockbl/.style ={rectangle, draw=black, thick, fill=white!20,
  text centered, rounded corners,
  minimum height=1em
},
blockb/.style ={rectangle, draw=blue, thick, fill=blue!20,
  text width=2em, text centered, rounded corners,
  minimum height=1em
},
blockr/.style ={rectangle, draw=red, thick, fill=red!20,
  text centered, rounded corners,
  minimum height=1em
},
blockg/.style ={rectangle, draw=gray!20, thick, fill=gray!20,
  text width=2em, text centered, rounded corners,
  minimum height=1em
},
blockw/.style ={rectangle, draw=blue, thick, fill=white!20,
  text centered, rounded corners,
  minimum height=1em
},
blockbw/.style ={rectangle, draw=blue, thick, fill=white!20,
  text centered, rounded corners,
  minimum height=10em
},
blocke/.style ={rectangle, draw=white, thick, fill=white!20,
  text centered, rounded corners,
  minimum height=1em
},
group/.style ={fill=gray!20, node distance=10mm},
ggroup/.style ={fill=red!5, node distance=10mm},
igroup/.style ={fill=red!15, node distance=10mm},
line/.style ={draw=black, -latex'},
lineb/.style ={draw=blue, -latex'},
liner/.style ={draw=red, -latex'},
lineg/.style ={draw=green, -latex'},
thickline/.style ={draw, thick, double, -latex', }
]

\node (r0) at (-1,5) [blocke] {$\rho_0$:};
\node (a0) at (2,4) [blockbl] {$i:Client$};
\node (b0) at (2,6) [blockbl] {$j: Proxy \wedge (<\; N\; C2P^-\; \top)$};

\node (A0) at (9,5) [blockw] {$del_E(i,j,Request); add_E(i,j,C2P)$};

\node (r1) at (-1,1) [blocke] {$\rho_1$:};
\node (a1) at (2,0) [blockbl] {$i:Client$};
\node (b1) at (2,2) [blockbl] {$j: Proxy \wedge (\geq\; N\; C2P^-\; \top)$};

\node (A1) at (9,1) [blockw] {$cl(j,k,\vv{\mathcal{L}_C});del_E(i,j,Request); add_E(i,k,C2P)$};

\begin{pgfonlayer}{background}
\node [group, fit=(a0) (b0)] (G0) {};
\node [group, fit=(a1) (b1)] (G1) {};
\end{pgfonlayer}

\path (a0) edge [-latex']                    node {$Request$} (b0);
\path (G0) edge [double,-latex]              node {} (A0);
\path (a1) edge [-latex']                    node {$Request$} (b1);
\path (G1) edge [double,-latex]              node {} (A1);

\end{tikzpicture}
}
\end{center}
\caption[Example of rule]{Example of rules used in \exampleref{ex:servernet}. In rule $\rho_1$, $\vv{\mathcal{L}_C} = \mathcal{R}_0 \backslash \{Request, C2P\}, \mathcal{R}_0 \backslash \{C2P\}, \emptyset, \emptyset, \emptyset$.}\label{fig:servernet}
\end{figure}


\section{Verification}\label{sec:verification}
Reasoning on graph transformations does not benefit yet from standard
 proof techniques as it is the case for term
rewriting. For instance, generalization of equational reasoning to graph
rewriting systems is not even complete \cite{CaferraEP08}.  In this
section, we follow a Hoare style to specify properties of \Grs 's for which
we establish a proof procedure.

\begin{definition}[Specification]
A \emph{specification} $SP$ is a triple $\{Pre\}(\mathcal{R},s)\{Post\}$ where $Pre$ and $Post$ are 
formulas (of a given logic), $\mathcal{R}$ is a graph rewriting
system and $s$ is a
strategy.
\end{definition}

\begin{definition}[Correctness]
A specification $SP$ is said to be \emph{correct} iff for all graphs
$G$, $G'$ such that $G \Rightarrow_{s} G'$ and $G \models
Pre$, then $G' \models Post$. 
\end{definition}

In order to show the correctness of a specification, we follow a
Hoare-calculus style  \cite{H69} and compute the weakest precondition
$wp(\mathcal{S},Post)$. For that, we give in \figref{fig:wp} (resp. in
\figref{fig:wpstrategy}) the definition of the function $wp$ which
yields the weakest precondition of a formula $Q$ w.r.t. an action
(resp. a strategy).

\begin{figure}[h]
\begin{framed}
\begin{tabular}{ll}
\hspace{-3mm}$wp (a,\; Q) = Q[a]$ \;\;\;\; & $wp (a;\alpha,\; Q) = wp(a, wp(\alpha,Q))$ 
\end{tabular}
\end{framed}
  \caption{Weakest preconditions w.r.t. actions where $a$ (resp. $\alpha$)
    stands for an elementary action (resp. action) and $Q$ is a formula.}\label{fig:wp}
\end{figure}

The weakest precondition of an elementary action, say $a$, and a
postcondition $Q$ is defined as $wp(a,Q) =Q[a]$ where $Q[a]$ stands
for the precondition consisting of $Q$ to which is applied a
substitution induced by the action $a$ that we denote by $[a]$. The
notion of substitution used here is the one coming from Hoare-calculi.

\begin{definition}[Substitutions]
  To each elementary action $a$ is associated a \emph{substitution},
  written $[a]$, such that for all graphs $G$ and formula $\phi$,
  ($G \models \phi[a]) \Leftrightarrow (G[a] \models \phi$).
\end{definition}
 
Notice that, in general, substitutions are not defined as formulas of
a given logic $\mathcal{L}$. They are defined as a new formula
constructor whose meaning is that the weakest preconditions for
elementary actions, as defined above, are correct.  In general, the
addition of a constructor for substitutions is not harmless. That is
to say, if $\phi$ is a formula of a logic $\mathcal{L}$, $\phi[a]$ is
not necessarily a formula of $\mathcal{L}$. It is a very interesting
problem to figure out which logics are closed under the considered
substitutions.  Some positive and negative answers are given in
Section~\ref{sec:logic}.

\begin{figure}[h]
\begin{framed}
\begin{tabular}{ll}
\hspace{-.3cm}$wp (\epsilon,\; Q) = Q$ &
\hspace{-.5cm}$wp (s_0 ; s_1,\; Q) = wp (s_0, wp(s_1,\; Q))$\\
\hspace{-.3cm}$wp (s_0 \oplus s_1,\; Q) = wp(s_0, Q) \wedge wp (s_1,Q)$&
\hspace{-.5cm}$wp (s^*,\; Q) = inv_s$ \\
\hspace{-.3cm}$wp (\rho,\; Q) = App(\rho) \Rightarrow wp(\alpha_\rho,Q)$ &
\hspace{-.5cm}$wp(\rho!,Q) = App(\rho) \wedge wp(\alpha_\rho,Q)$\\
\hspace{-.3cm}$wp(\rho?,Q) = (App(\rho) \Rightarrow wp(\alpha_\rho,Q)) \wedge (\neg App(\rho) \Rightarrow Q) $
\end{tabular}
\end{framed}
  \caption{Weakest preconditions for strategies. $\alpha_\rho$ denotes the right-hand side of rule $\rho$.}\label{fig:wpstrategy}
\end{figure}

The definition of $wp (s, Q)$ for the empty strategy, the composition
and the choice are quite direct.  The definitions for the rule,
mandatory rule and trial differ on what happens if the rule cannot be
applied.  When the rule $\rho$ can be applied, then applying it should
lead to a graph satisfying $Q$.  When the rule $\rho$ cannot be
applied, $wp(\rho,\; Q)$ indicates that the considered specification
is correct; while $wp(\rho!,\; Q)$ indicates that the specification is
not correct and $wp(\rho?,\; Q)$ leaves the postcondition unchanged
and thus transformations can move to possible next steps.

The weakest precondition for the closure is close to the $while$
imperative instruction. It requires an invariant $inv_s$ to be defined.
$wp (s^*,\; Q) = inv_s$ which means that the invariant has to be true
when entering the iteration for the first time. On the other hand, it
is obviously  not enough to be sure that $Q$ will be satisfied when
exiting the iteration or that the invariant will be maintained
throughout execution. To make sure that iterations behave correctly,
we need to introduce some additional \emph{verification conditions}
computed by means of a function $vc$, defined in  \figref{fig:vcstrategy}.

\begin{figure}[h]
\begin{framed}
\begin{tabular}{lll}
\hspace{-.3cm}$vc(\epsilon,\; Q)$ &$=$& $\top$  (true)\\
\hspace{-.3cm}$vc(s_0 ; s_1,\; Q)$ &$=$& $vc(s_0, wp(s_1,\; Q)) \wedge vc(s_1,Q)$\\
\hspace{-.3cm}$vc(s_0 \oplus s_1,\; Q)$ &$=$& $vc(s_0, Q) \wedge vc(s_1,Q)$\\
\hspace{-3mm}$vc(s^*,\; Q)$ &$=$& $vc(s,Q) \wedge (inv_s \wedge App(s)
                                  \Rightarrow wp(s,inv_s)) \wedge (inv_s \wedge \neg App(s) \Rightarrow Q)$ \\
\hspace{-.3cm}$vc(\rho,\; Q)$ &$=$& $\top$ \\
\hspace{-.3cm}$vc(\rho!,Q)$ &$=$& $\top$\\
\hspace{-.3cm}$vc(\rho?,Q)$ &$=$& $\top$
\end{tabular}
\end{framed}
  \caption{Verification conditions for strategies. }\label{fig:vcstrategy}
\end{figure}

As the computation of $wp$ and $vc$ requires the user to provide invariants, we now introduce the notion of annotated strategies and specification.

\begin{definition}[Annotated strategy, Annotated specification]
An \emph{annotated strategy} is a strategy in which every iteration $s^*$ is annotated with an invariant $inv_s$. It is written $s^*\{inv_s\}$.
An \emph{annotated specification} is a specification whose strategy is an annotated strategy.
\end{definition}

\begin{definition}[Correctness formula]
We call \emph{correctness formula} of an annotated specification $SP =
\{Pre\}(\mathcal{R},s)\{Post\}$, the formula~: $$correct(SP) = (Pre \Rightarrow wp(s,Post)) \wedge vc(s,Post).$$
\end{definition}

Before stating the soundness of the proposed verification method, we
state a first simple lemma.

\begin{lemma}\label{cor:sound}
  Let
  $Q$ be a formula and $\alpha$ be an action. For all graphs $G$, $G'$ such that
  $G \rightarrow_{\alpha} G'$, $G \models wp(\alpha,Q)$ implies
  $G' \models Q$.
\end{lemma}

\begin{proof}
  \begin{itemize}
\item Let us assume $\alpha = a$, an elementary action. Then $wp(a,\; Q) = Q[a]$. Let
  $G$ be a graph such that $G \models Q[a]$. By definition of the substitutions,
  $G \models Q[a]$ implies that for any graph $G'$
  such that $G \rightarrow_{a} G'$, $G' \models Q$. Thus $G \models wp(\alpha,\; Q) \Rightarrow G' \models Q$.
\item Let us assume $\alpha = a;\; \alpha'$ where $a$ is an elementary action and $\alpha'$ is an action. Then $wp(a;\; \alpha',\; Q) = wp(a, wp(\alpha', Q))$. Let $G$ be a graph such that
  $G \models wp(a, wp(\alpha', Q))$ and let
  $G'$ be a state such that $G \Rightarrow_{a; \alpha'}
  G'$. Then there exists $G''$ with $G
  \rightarrow_{a} G''$ and $G'' \Rightarrow_{\alpha'}
  G'$. As $G \models wp(a, wp(\alpha', Q))$, by induction,
  $G''\models wp(\alpha',Q)$. Then, by an additional induction, $G' \models Q$. Thus $G \models wp(\alpha,\; Q) \Rightarrow G' \models Q$.
\end{itemize}
\end{proof}

\begin{theorem}[Soundness]\label{th:sound}
  Let
  $SP = \{Pre\}(\mathcal{R}, s)\{Post\}$ be an annotated specification. If
   $correct(SP)$ is valid,
  then for all graphs $G$, $G'$ such that
  $G \Rightarrow_{s} G'$, $G \models Pre$ implies
  $G' \models Post$.
\end{theorem}

\begin{proof}
This proof is done by induction on the semantic of the programming
language.
\begin{itemize}
\item Let us assume $s = \rho$. Then $correct(SP) = Pre \Rightarrow wp(\rho,\; Post)$. Let $G$, $G'$ be graphs such that $G \models Pre$ and $G \Rightarrow_{\rho} G'$ then, as $correct(SP)$ is valid, $G \models Pre \Rightarrow wp(\rho,\; Post)$. Thus, by modus ponens, $G \models wp(\rho,\; Post)$. As $wp(\rho,\; Post) = App(\rho) \Rightarrow wp(\alpha_\rho,\, Post)$ and, by definition of $\Rightarrow_{\rho}$, $G \models App(\rho)$  and thus, by modus ponens, $G \models wp(\alpha_\rho,\; Post)$. Then, by applying the lemma, $G' \models Post$.
\item Let us assume $s = \rho!$. Then $correct(SP) = Pre \Rightarrow
  wp(\rho!, Post)$. Let $G$, $G'$ be graphs such that $G \models
  Pre$ and $G \Rightarrow_{\rho!} G'$ then, as $correct(SP)$ is valid,
  $G \models Pre \Rightarrow wp(\rho!,\; Post)$. Thus, by modus
  ponens, $G \models wp(\rho!,\; Post)$. As $wp(\rho!,\; Post) =
  App(\rho) \wedge wp(\alpha_\rho,\, Post)$, $G \models
  wp(\alpha_\rho,\; Post)$. Then, by applying the lemma, $G'
  \models Post$.
\item Let us assume $s = \rho?$. Then $correct(SP) = Pre \Rightarrow wp(\rho?,\; Post)$. Let $G$, $G'$ be graphs such that $G \models Pre$ and $G \Rightarrow_{\rho?} G'$ then, as $correct(SP)$ is valid, $G \models Pre \Rightarrow wp(\rho?,\; Post)$. Thus, by modus ponens, $G \models wp(\rho?,\; Post)$. As $wp(\rho?,\; Post) = (App(\rho) \Rightarrow wp(\alpha_\rho,\, Post)) \wedge (\neg App(\rho) \Rightarrow Post)$, we have to treat two different cases:
  \begin{itemize}
  \item if $G \models App(\rho)$ then, by modus ponens, $G \models wp(\alpha_{rho},Post)$ and then, as, by definition of $\Rightarrow_{\rho?}$, $G \Rightarrow_{\rho} G'$, using the lemma, $G' \models Post$.
  \item otherwise, $G \models \neg App(\rho)$ and thus, by modus ponens, $G \models Post$. But, by definition of $\Rightarrow_{\rho?}$, $G = G'$ and thus $G' \models Post$.
  \end{itemize}
  Thus $G \models Pre \Rightarrow G' \models Post$.
\item Let us assume $s = s_0;\; s_1$. Then $correct(SP) = vc(s_0;
  s_1,\; Post)\; \wedge\; (Pre \Rightarrow wp(s_0;\; s_1,\; Post)$. As
  $vc(s_0;\; s_1,\; Post) = vc(s_0, wp(s_1, Post)) \wedge vc(s_1,
  Post)$ and $wp(s_0;\; s_1,\; Post) = wp(s_0, wp(s_1, Post))$,
  $correct(SP) = vc(s_0,\\ wp(s_1, Post)) \wedge vc(s_1, Post) \wedge
  (Pre \Rightarrow wp(s_0, wp(s_1, Post))$. Let $G$ be a graph such that
  $G \models Pre$. As $correct(SP)$ is valid, $G \models correct(SP)$. Let
  $G'$ be a graph such that $G \Rightarrow_{s_0; s_1}
  G'$. Then there exists $G''$ with $G
  \Rightarrow_{s_0} G''$ and $G'' \Rightarrow_{s_1}
  G'$. As $G \models Pre$ and $G \models vc(s_0,wp(s_1,Post) \wedge
  (Pre \Rightarrow wp(s_0,wp(s_1,Post)))$, by induction with $S_0 =
  (Pre, wp(s_1,Post), \mathcal{R}, s_0)$,
  $G'' \models wp(s_1,Post)$. As $correct(SP)$ is valid, so is $vc(s_1,Post)$ and thus also $vc(s_1,Post) \wedge
  (wp(s_1,Post) \Rightarrow wp(s_1,Post))$. Once more, by induction with $S_1 = (wp(s_1,Post), Post, \mathcal{R}, s_1)$, $G' \models Post$. Thus $G \models Pre \Rightarrow G' \models Post$.
\item Let us assume that $s = \epsilon$. Then $correct(SP) = Pre \Rightarrow Post$. Let $G$ and $G'$ be graphs such that $G \Rightarrow_{\epsilon} G'$ and $G \models Pre$. By definition, $G = G'$ and thus, by modus ponens, $G' \models Post$. Thus $G \models Pre \Rightarrow G' \models Post$.
\item Let us assume $s = s_0 \oplus s_1$. Then $correct(SP) = vc(s_0,Post)\; \wedge\; vc(s_1,Post)\; \wedge\; (Pre
  \Rightarrow wp(s_0,\; Post) \wedge wp(s_1,\; Post))$. Let $G$ and $G'$ be graphs such that $G \Rightarrow_{s_0 \oplus s_1} G'$ and $G \models Pre$. By definition of $\Rightarrow_{s_0 \oplus s_1}$, there are two possible cases:
\begin{itemize}
\item If $G \Rightarrow_{s_0} G'$ then, as $correct(SP)$ is valid, so is $vc(s_0,\\ Post) \wedge (Pre \Rightarrow wp(s_0,\; Post))$. As $G \models Pre$, by induction, $G' \models Post$.
\item otherwise, $G \Rightarrow_{s_1} G'$ and then, as $correct(SP)$ is valid, $vc(s_1,\; Post) \wedge (Pre \Rightarrow wp(s_1,Post))$. As $G \models Pre$, by induction, $G' \models Post$.
\end{itemize}
Thus $G \models Pre \Rightarrow G' \models Post$.
\item Let us assume $s = s_0^*\{inv\} $. Then $correct(SP) = vc(s_0,inv) \wedge (inv \wedge App(s_0) \Rightarrow wp(s_0,inv)) \wedge (inv \wedge \neg App(s_0) \Rightarrow Post) \wedge \; (Pre \Rightarrow inv)$. Let $G$ and $G'$ be graphs such that $G \models Pre$ and $G \Rightarrow_{s_0^*} G'$. There are two possible cases: 
\begin{itemize}
\item Let us assume that $G \models App(s_0)$. By definition of $\Rightarrow_{s_0^*}$, there exist $G''$ such that $G \Rightarrow_{s_0} G''$ and $G'' \Rightarrow_{s_0^*} G'$. As $correct(SP)$ is valid, so is $vc(s_0,inv) \wedge (Pre
  \Rightarrow (App(s_0) \wedge inv \Rightarrow wp(s_0,inv)))$. By induction with $S' = (Pre \wedge App(s_0)
  \wedge inv,inv,\mathcal{R},s_0)$, as $G \models Pre \wedge App(s_0)
  \wedge inv$, $G'' \models inv$. Similarly, with $S'' = (inv, Post, \mathcal{R}, s_0^*\{inv\})$, by induction, $G' \models Post$.
\item otherwise, $G \models \neg App(s_0)$ and thus, by modus ponens, $G \models Post$. But, by definition of $\Rightarrow_{s_0^*}$, $G = G'$ and thus $G' \models Post$.
\end{itemize}
Thus $G \models Pre \Rightarrow G' \models Post$.
\end{itemize}
\end{proof}

\begin{example}\label{ex:servernetcorr}
Let us consider \exampleref{ex:servernet}. We want to prove that the
specification $\{Pre\} (\mathcal{R}, \rho_0 \oplus \rho_1) \{Post\} $, where $Pre \equiv \exists U. (Client \wedge \exists Request. Proxy) \wedge \forall U. (Proxy \Rightarrow (\le\; N\; C2P\; \top))$ and $Post = \forall U. (Proxy \Rightarrow (\le\; N\; C2P\; \top))$, is correct. $Pre$ means that there exist a $Client$ that $Request$ed a connection to a $Proxy$ and that no $Proxy$ has more than $N$ different $C2P$ incoming connections. $Post$ means that no $Proxy$ has more than $N$ different $C2P$ incoming connections.
The correctness formula is then $(Pre \Rightarrow wp(\rho_0 \oplus \rho_1,Post)) \wedge vc(\rho_0 \oplus \rho_1,Post)$. It can be simplified, however, as $vc(\rho_0 \oplus \rho_1,Post) =  vc(\rho_0,Post) \wedge vc(\rho_1,Post)$ and both $vc(\rho_0,Post)$ and $vc(\rho_1,Post)$ are, by definition, true. The correctness formula is thus $Pre \Rightarrow ((App(\rho_0) \Rightarrow wp(\alpha_{\rho_0},Post)) \wedge (App(\rho_1) \Rightarrow wp(\alpha_{\rho_1},Post)))$
\end{example}


\section{Assertion Logics}\label{sec:logic}
The framework presented so far regarding the considered rewrite
systems (LDRSs) and specifications is parameterized by a given logic
$\mathcal{L}$. In this section, we present some logics that could
possibly be used to instantiate this general framework. The logics
that are used should be closed under the substitutions generated by
the elementary actions. Otherwise, the computation of weakest
preconditions may be outside the considered logic.  We start by considering
first-order logic as well as some of its decidable fragments. We
focus more particularly on description logics (DL) in a second time. We show that some
of them are closed under substitutions for all the actions that we
have presented in this paper. We also provide a negative result by
proving that some DL fragments are not closed under substitutions
generated by the elementary action \emph{merge} ($mrg(i,j)$). The
results presented in this section are new and complete those already
given in \cite{BESDL16}. For all the logics we consider, we discuss the closure under substitutions and the expression of the literal $App(\rho)$. 

\subsection{First-order logic}\label{ssec:FO}

We start by recalling briefly the first-order formulas useful for our
purpose as well as the notions of interpretations and models.

\begin{definition}[First-order formula]
  Let $\mathcal{A} = (\mathcal{V}, \mathcal{C},
  \mathcal{R})$ where $\mathcal{V}$ is a set of variables, 
  $\mathcal{C}$ is a set of unary predicates, and $\mathcal{R}$ is a set of  binary
  predicates including equality ( $=$ ). Given $x, y \in \mathcal{V}$, $C \in \mathcal{C}$ and $R
  \in \mathcal{R}$, the set of first-order formulas $\phi$ we consider is defined by:\\
  \begin{tabular}{lll|l|l|l}
$\phi :=\;\top\; |\; C(x)\; |\; R(x,y)\; |\; x = y\; |\; \neg \phi\; |\; \phi \vee \phi\; |\; \exists x. \phi$
\end{tabular}\\

  For the sake of conciseness, we define $\bot \equiv \neg \top$, $\phi \wedge \psi \equiv \neg (\neg \phi \vee \neg \psi)$, $\forall x. \phi \equiv \neg (\exists x. \neg \phi)$.
  
  A variable $x$ is free in $\phi$ iff $\phi = C(t_0)$, $\phi =
  R(t_0,t_1)$ or $\phi = " t_0 = t_1"$ and $x$ occurs in $t_0$ or
  $t_1$, or $\phi = \neg \psi$ or $\phi = \psi \vee \psi'$ and $x$ is free in $\psi$ and $\psi'$, or $\phi = \exists y. \psi$ and $x$ is free in $\psi$ and $x$ is different from $y$.
  A formula with no free variable is a sentence. We only consider sentences hereafter.
\end{definition}

\begin{definition}[Model]
  Let $G = (N, E, \phin, \phir, s, t)$ be a graph over the alphabet
  $(\mathcal{C},\mathcal{R})$, an interpretation over the alphabet
  $(\mathcal{V}, \mathcal{C}, \mathcal{R})$ is a tuple $(\Delta,
  \cdot^{\mathcal{I}})$ such that $N \subseteq \Delta$ and
  $\cdot^\mathcal{I}$ is a function over formulas defined by:
  \begin{itemize}
  \item $\top^{\mathcal{I}}$ is true
  \item $C(x)^{\mathcal{I}}$ is true if and only if $C \in \phin(x)$
  \item $R(x,y)^{\mathcal{I}}$ is true if and only if $\exists e \in E. s(e) = x$ and $t(e) = y$ and $R \in \phir(e)$
  \item $x =^\mathcal{I} y$ is true if and only if $x$ is $y$
  \item $(\exists x. \phi)^\mathcal{I}$ is true if and only if $\exists n \in N. \phi[x \rightarrow n]^{\mathcal{I}}$ where $\phi[x \rightarrow n]$ is $\phi$ where each occurrence of $x$ is replaced with $n$
  \item $(\neg \phi)^{\mathcal{I}}$ is true if and only if not $\phi^{\mathcal{I}}$
  \item $(\phi \vee \psi)^{\mathcal{I}}$ is true if and only if $\phi^{\mathcal{I}}$ or $\psi^{\mathcal{I}}$
  \end{itemize}
  
  We say that a graph $G$ models a first-order formula $\phi$, written
  $G \models \phi$ if there exists an interpretation $(\Delta,
  \cdot^{\mathcal{I}})$ such that $\phi^\mathcal{I}$ is true.
\end{definition}

One may remark that $N \subseteq \Delta$ and not $N = \Delta$. This is
because some actions (e.g. node addition, deletion, merging ...)  may
modify the set of nodes currently existing (i.e.  nodes of the current
graph). To keep track of $N$, we follow \cite{BESICTAC16} and
introduce a special unary predicate $Active$ that denotes the existing
nodes. We transform formulas so that all $\exists x. \phi$ become
$\exists x. Active(x) \wedge \phi$ and add to the definition of
$\cdot^{\mathcal{I}}$ the fact that $Active(x)^{\mathcal{I}}$ is true
if and only if $x \in N$ and $\exists x. \phi$ if and only if
$\exists n \in \Delta. \phi[x \rightarrow n]$ where
$\phi[x \rightarrow n]$ is $\phi$ where each occurrence of $x$ is
replaced with $n$.

\begin{theorem}
First-order logic is closed under substitutions.
\end{theorem}

\begin{proof}
The proof is done by induction on the formula constructors. We focus here on the substitutions generated by the elementary actions $mrg$ anc $cl$. The full proof is reported in the appendix.

We start by giving formulas without substitutions that are equivalent to those with substitutions.
  \begin{itemize}
\item $\top[\sigma] \leadsto \top$
\item $Active(x)[\sigma] \leadsto Active(x)$ if $\sigma \neq add_N(i)$
  and $\sigma \neq  cl(i,j,\dots)$ and $\sigma \neq mrg(i,j)$
\item $Active(x)[add_N(i)] \leadsto Active(x) \vee i = x$
\item $Active(x)[del_N(i)] \leadsto Active(x) \wedge i \neq x$
  \item $Active(x)[cl(i,j,\dots)] \leadsto Active(x) \vee x = j$
  \item $Active(x)[mrg(i,j)] \leadsto Active(x) \wedge x \neq j$  
\item $C'(x)[add_C(i,C)] \leadsto C'(x)$
\item $C(x)[add_C(i,C)] \leadsto C(x) \vee i = x$
\item $C'(x)[del_C(i,C)] \leadsto C'(x)$
\item $C(x)[del_C(i,C)] \leadsto C(x) \wedge \neg i = x$
\item $C(x)[add_R(i,j,R)] \leadsto C(x)$
\item $C(x)[del_R(i,j,R)] \leadsto C(x)$
\item $C(x)[add_N(i)] \leadsto C(x)$ for $C \neq Active$
\item $C(x)[del_N(i)] \leadsto C(x) \wedge \neg i = x$ for $C \neq
  Active$
\item $C(x)[i \gg j] \leadsto C(x)$
  \item $C(x)[cl(i,j,\dots)] \leadsto C(x) \vee (x = j \wedge C(i))$ if $C \neq Active$
  \item $C(x)[mrg(i,j)] \leadsto x \neq j \wedge (C(x) \vee (x = i
    \wedge C(j)))$ if $C \neq Active$
\item $R(x,y)[add_C(i,C)] \leadsto R(x,y)$
\item $R(x,y)[del_C(i,C)] \leadsto R(x,y)$
\item $R'(x,y)[add_R(i,j,R)] \leadsto R'(x,y)$
\item $R(x,y)[add_R(i,j,R)] \leadsto R(x,y) \vee (i = x \wedge j = y)$
\item $R'(x,y)[del_R(i,j,R)] \leadsto R'(x,y)$
\item $R(x,y)[del_R(i,j,R)] \leadsto R(x,y) \wedge (\neg i = x \vee
  \neg j = y)$
\item $R(x,y)[add_N(i)] \leadsto R(x,y)$
\item $R(x,y)[del_N(i)] \leadsto R(x,y) \wedge \neg i = x \wedge \neg i = y$
\item $R(x,y)[i \gg j] \leadsto (R(x,y) \wedge \neg i = y) \vee
  (R(x,i) \wedge j = y))$ 
  \item $R(x,y) [cl(i,j,\dots)] \leadsto R(x,y) \vee \phiin \vee \phiout \vee \philin \vee \philout \vee \phill$ where:
    \begin{itemize}
  \item $\phiin = \left\{
	\begin{array}{ll}
	  R(x,i) \wedge y = j  \wedge \neg (x = i) & \mbox{if } R \in \rin\\
          \bot  & \mbox{otherwise }
	\end{array}
        \right.$
  \item $\phiout = \left\{
	\begin{array}{ll}
	  R(i,y) \wedge x = j  \wedge \neg (y = i) & \mbox{if } R \in \rout\\
          \bot  & \mbox{otherwise }
	\end{array}\right.$
  \item $\philin = \left\{
	\begin{array}{ll}
	  R(i,i) \wedge x = i \wedge y = j  & \mbox{if } R \in \rlin\\
          \bot  & \mbox{otherwise }
	\end{array}
        \right.$
  \item $\philout = \left\{
	\begin{array}{ll}
	  R(i,i) \wedge x = j \wedge y = i  & \mbox{if } R \in \rlout\\
          \bot  & \mbox{otherwise }
	\end{array}
        \right.$
  \item $\phill = \left\{
	\begin{array}{ll}
	  R(i,i) \wedge x = j \wedge y = j  & \mbox{if } R \in \rll\\
          \bot  & \mbox{otherwise }
	\end{array}
        \right.$
    \end{itemize}
  \item $R(x,y) [mrg(i,j)] \leadsto x \neq j \wedge y \neq j \wedge (R(x,y) \vee (R(x,j) \wedge y = i) \vee$\\$ (R(j,y) \wedge x = i) \vee (x = i \wedge y = i \wedge R(j,j)))$  
\item $(\exists x. \phi)[\sigma]) \leadsto \exists x. (\phi[\sigma])$
\item $(\phi \wedge \psi)[\sigma] \leadsto \phi[\sigma] \wedge
  \psi[\sigma]$
\item $(\neg \phi)[\sigma] \leadsto \neg (\phi[\sigma])$
  \end{itemize}

Let us now prove that the proposed formulas without
  substitutions are indeed euivalent to the ones with
  substitutions. For lack of space, we will illustrate these equivalences only for some of them. To do that, we introduce the interpretations ($\Delta^G,\cdot^G$) and ($\Delta^{G'},\cdot^{G'}$) that results from the cloning or merging action.
\begin{description}
\item[{$\top[\sigma]$:}] No matter what action is
  performed, $\top$ is satisfied.
\item[{$Active(x)[\sigma]$:}] If $\sigma$ is not a node creation,
  deletion, cloning or merging all nodes that were active stay so and vice-versa.
\item[{$Active(x)[add_N(i)]$:}] The
  valuation of $Active$ becomes $Active^G \cup \{i^G\}$.
\item[{$Active(x)[del_N(i)]$:}] The
  valuation of $Active$ becomes $Active^G \backslash \{i^G\}$.
\item[{$Active(x)[cl(i,j,\dots)]$:}] As $N^{G'} = N^G \cup {j}$,
  $Active(x) \in \phi_N^{G'}(n)$ if and only if $(Active(x) \vee x =
  j) \in \phi_N^G(n)$.
  \item[{$Active(x)[mrg(i,j)]$:}] As $N^{G'} = N^G \backslash {j}$,
    $Active(x)^{G'}$ if and only if $(Active(x) \wedge x \neq j)^G$.
\item[{$C'(x)[add_C(i,C)]$:}] The valuation of $C'$ is left
  untouched.
\item[{$C(x)[add_C(i,C)]$:}] $C^{G'}$ after
  performing $add_C(i,C)$ is $C^{G} \cup \{i^{G}\}$.
\item[{$C'(x)[del_C(i,C)]$:}] The valuation of $C'$ is left
  untouched.
\item[{$C(x)[add_R(i,j,R)]$:}] The valuation of $C$ is left
  untouched.
\item[{$C(x)[del_R(i,j,R)]$:}] The valuation of $C$ is left
  untouched.
\item[{$C(x)[del_C(i,C)]$:}] $C^{G'}$ after
  performing $del_C(i,C)$ is $C^{G} \backslash \{i^{G}\}$.
\item[{$C(x)[add_N(i)]$:}] The valuation of $C$ is left
  untouched. 
\item[{$C(x)[del_N(i)]$:}] $C^{G'} = C^{G}
  \backslash \{i^{G}\}$.
\item[{$C(x)[i \gg j]$:}] $C^{G'} = C^{G}$.
  \item[{$C(x)[cl(i,j,\dots)]$:}] As $\phi_N^{G'}(j) = \phi_N^G(i)$ and $\forall n \neq j, \phi_N^{G'}(n) = \phi_N^G(n)$, $C(x)^{G'}$ if and only if $(C(x) \vee (x = j \wedge C(i))^G$.
  \item[{$C(x)[mrg(i,j)]$:}] As $\phi_N^{G'}(i) = \phi_N^G(i) \cup
    \phi_N^G(j)$ and $\forall n \neq j, \phi_N^{G'}(n) = \phi_N^G(n)$,
    $C(x)^{G'}$ if and only if $(x \neq j \wedge (C(x) \vee (x = i
    \wedge C(j))))^G$.
\item[{$R(x,y)[add_C(i,C)]$:}] The valuation of $R$ is left
  untouched. 
\item[{$R(x,y)[del_C(i,C)]$:}] The valuation of $R$ is left
  untouched.
\item[{$R'(x,y)[add_R(i,j,R)]$:}] The valuation of $R'$ is left
  untouched.
\item[{$R(x,y)[add_R(i,j,R)]$:}] $R^{G'}$ is $R^{G} \cup \{(i^{G},j^{G})\}$. 
\item[{$R'(x,y)[del_R(i,j,R)]$:}] The valuation of $R'$ is left
  untouched.
\item[{$R(x,y)[del_R(i,j,R)]$:}] $R^{G'}$ is $R^{G} \backslash \{(i^{G},j^{G})\}$.
\item[{$R(x,y)[add_N(i)]$:}] The valuation of $R$ is left
  untouched. 
\item[{$R(x,y)[del_N(i)]$:}] $R^{G'} = R^{G}
  \backslash \{(a,b)| a \in i^{G} \text{ or } b \in i^{G}\}$.
\item[{$R(x,y)[i \gg j]$:}] $R^{G'} = R^{G} \cup \{(a,j) | (a,i) \in
  R^{G}\} \backslash \{(a,i) \in R^{G}\}$. 
 \item[{$R(x,y)[cl(i,j,\dots)]$:}] $R(x,y) [cl(i,j,\dots)] \leadsto R(x,y) \vee \phiin \vee \phiout \vee \philin \vee \philout \vee \phill$:\\
  If $R \in \phi_E^{G'}(e')$ then either:
    \begin{itemize}
    \item $e' \in \ein$ and then $x = s^{G}(in(e'))$ and $y = j$, that is there exists $e$ such that $R \in \phi_E^G(e)$ and $s^G(e) = x$ and $t^G(e) = i$. Thus $(R(x,i) \wedge y = j \wedge \neg (x = i))^G$.
    \item $e' \in \eout$ and then $x = j$  and $y = t^{G}(out(e'))$, that is there exists $e$ such that $R \in \phi_E^G(e)$ and $s^G(e) = i$ and $t^G(e) = y$. Thus $(R(i,y) \wedge x = j \wedge \neg (y = i))^G$.
    \item $e' \in \elin$ and then $x = i$, $y = j$ and there exists $e$ such that $R \in \phi_E^G(e)$ and $s^G(e) = i$ and $t^G(e) = i$. Thus $(R(i,i) \wedge x = i \wedge y = j)^G$.
    \item $e' \in \elout$ and then $x = j$, $y = i$ and there exists $e$ such that $R \in \phi_E^G(e)$ and $s^G(e) = i$ and $t^G(e) = i$. Thus $(R(i,i) \wedge x = j \wedge y = i)^G$.
    \item $e' \in \ell$ and then $x = j$, $y = j$ and there exists $e$ such that $R \in \phi_E^G(e)$ and $s^G(e) = i$ and $t^G(e) = i$. Thus $(R(i,i) \wedge x = j \wedge y = j)^G$. 
    \item otherwise, $e' \in E^G$ and thus $R(x,y)^G$.
    \end{itemize}
  \item[{$R(x,y)[mrg(i,j)]$:}] $R^{G'} = \{(x,y) | x \neq j \text{ and } y \neq j \text{ and } (R(x,y) \text{ or } (R(x,j) \text{ and } y = i) \text{ or } (R(j,j) \text{ and } y = i \text{ and } x = i))\}$ 
\item[{$(\exists x. \phi)[\sigma]$:}]
  The substitutions do not modify the existence or not of a node.
\item[{$(\phi \wedge \psi)[\sigma]$:}] If $\phi \wedge \psi$ is satisfied after performing
  $\sigma$, so must be $\phi$ and $\psi$ and the other way round.
\item[{$(\neg \phi)[\sigma]$:}] If $\phi$
  is not satisfied after performing
  $\sigma$, it is not possible that $\phi$ be satisfied after
  performing $\sigma$.
  \end{description}
\end{proof}

The proof of the previous theorem shows that the shape of the used formulas is conserved. If the formulas (ignoring the substitutions) belonged to less expressive and decidable fragments of first-order logic, namely the two-variable fragment with counting $\mathcal{C}2$\cite{graedel_otto_rosen_lics97} and $\forall^*\exists^*$, the fragment containing only formulas that, in prenex normal form, can be written as $\forall x_0\dots\forall x_n\exists y_0\dots\exists y_n. \phi$ with $\phi$ quantifier free\cite{boerger_graedel_gurevich_decision_problem}, so do the equivalent formulas without substitution.

\begin{corollary}
$\forall^*\exists^*$ and $\mathcal{C}2$ are closed under substitutions.
\end{corollary}

The correctness formula includes substitutions as well as literals of the form $App(\rho)$. We proved previously that it is possible to remove the substitutions. Below, we show that $App(\rho)$ can be expressed in first-order logic. 

\begin{proposition}
Let us assume that $\rho$ is a rule such that the labels of its
left-hand side are in first-order logic. It is possible to express
$App(\rho)$ in first order logic.
\end{proposition}

\begin{proof}
Let $L = (N^L, E^L, \phin^L, \phir^L, s^L, t^L)$ be the left-hand side of
$\rho$. Let $A = \exists_{n \in N^L} x_n. \bigwedge_{n \in
  N^L} \psi_n \wedge
\bigwedge_{e \in E^L}  \psi_e$\footnote{$\exists_{n \in N} x_n$ is used
  as a shorthand for $\exists x_{n_0}. \dots \exists x_{n_k}$ where $N
  = \{n_0, \dots, n_k\}$.} where $\psi_n = \bigwedge_{c \in \phin^L(n)} c(x_n)$
and\\ $\psi_e = \bigwedge_{r \in \phir^L(e)} r(x_{s^L(e)},x_{t^L(e)})$.

Let us assume that $G = (N^G, E^G, \phin^G, \phir^G, s^G, t^G)$ is
a graph. 

Let us assume that $G \models A$. Then, let us define $h^N(n) = x_n$,
for $n \in N^L$, and $h^E(e) = \xi_e$ where $\xi_e \in \{e' \in E^G |
s^G(e') = x_{s^L(e)} \wedge t^G(e') = x_{t^L(e)}\}$, for $e \in E^L$.
\begin{enumerate}
\item For all $n \in N^L$, for all $c \in \phin^L(n)$, $x_n \models c$
\item For all $e \in E^L$, for all $r \in \phir^L(e)$, $\xi_{e}
  \models r$
\item  For all $e \in E^L$, $s^G(\xi_e) = x_{s^L(e)}$
\item  For all $e \in E^L$, $t^G(\xi_e) = x_{t^L(e)}$
\end{enumerate}
$(h^N,h^E)$ is thus a match. Hence, there exist at least one.

Let us now assume that there exists a match $(h^N, h^E)$ from $L$ to
$G$. Then, by definition, the $x_n= h^N(n)$'s (and $\xi_e=h^E(e)$'s) of $A$
exist. Additionally, due to the first condition, $x_n$ is a model of
$\psi_n$ and, thanks to the other conditions, $\xi_e$ is a model of
$\psi_e$. Thus $G \models A$.

Thus, $A \Leftrightarrow App(\rho)$.   
\end{proof}

\begin{example}
Let us consider the rule $\rho$ of \figref{fig:NAppFO}. The corresponding $App(\rho)$ in first-order logic is $\exists i,j,k. C(j) \wedge (C(k) \vee D(k)) \wedge R(i,j) \wedge R(j,k)$. This formula is in $\forall^*\exists^*$ but not in $\mathcal{C}^2$. However, it is equivalent to the $\mathcal{C}^2$ formula $\exists x,y. (R(x,y) \wedge C(y) \wedge \exists x.((C(x) \vee D(x)) \wedge R(y,x)))$.
\end{example}

\begin{figure}
\def\smallscale{0.6}
\begin{center}
\resizebox{10cm}{!}{
\begin{tikzpicture}
[auto,
blockbl/.style ={rectangle, draw=black, thick, fill=white!20,
  text centered, rounded corners,
  minimum height=1em
},
blockb/.style ={rectangle, draw=blue, thick, fill=blue!20,
  text width=2em, text centered, rounded corners,
  minimum height=1em
},
blockr/.style ={rectangle, draw=red, thick, fill=red!20,
  text centered, rounded corners,
  minimum height=1em
},
blockg/.style ={rectangle, draw=gray!20, thick, fill=gray!20,
  text width=2em, text centered, rounded corners,
  minimum height=1em
},
blockw/.style ={rectangle, draw=blue, thick, fill=white!20,
  text centered, rounded corners,
  minimum height=1em
},
blockbw/.style ={rectangle, draw=blue, thick, fill=white!20,
  text centered, rounded corners,
  minimum height=10em
},
blocke/.style ={rectangle, draw=white, thick, fill=white!20,
  text centered, rounded corners,
  minimum height=1em
},
group/.style ={fill=gray!20, node distance=10mm},
ggroup/.style ={fill=red!5, node distance=10mm},
igroup/.style ={fill=red!15, node distance=10mm},
line/.style ={draw=black, -latex'},
lineb/.style ={draw=blue, -latex'},
liner/.style ={draw=red, -latex'},
lineg/.style ={draw=green, -latex'},
thickline/.style ={draw, thick, double, -latex', }
]

\node (r0) at (0,5) [blocke] {$\rho$:};
\node (a0) at (2,4) [blockbl] {$i$};
\node (b0) at (3,6) [blockbl] {$j: C$};
\node (c0) at (4,4) [blockbl] {$k: C \vee D$};

\node (A0) at (9,5) [blockw] {$add_C(i,A);$};

\begin{pgfonlayer}{background}
\node [group, fit=(a0) (b0) (c0)] (G0) {};
\end{pgfonlayer}

\path (a0) edge [-latex']                    node {$R$} (b0);
\path (b0) edge [-latex']                    node {$R$} (c0);

\path (G0) edge [double,-latex]              node {} (A0);

\end{tikzpicture}
}
\end{center}
\caption{Example of a rule with labels in first-order logic}
\label{fig:NAppFO}
\end{figure}
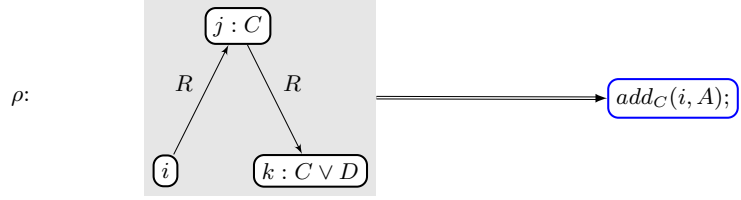


\subsection{Description logics}\label{ssec:DL}

 Description Logics are also fragments of first-order logic but not all
 description logics are closed under substitutions. We mainly focus in this
 subsection on  the substitutions generated by the cloning and merging
 elementary actions. Closure under classical substitutions have been
 considered in \cite{BESDL16}. We prove that with the addition of merge and global
 edge redirection, some logics are still closed while others no longer are.

We assume that the reader is familiar with Description Logics (see \cite{DLHandbook} for extended definitions). We only focus on extensions of $\mathcal{ALC}$. We recall that these extensions are named by appending a letter representing additional constructors to the logic name. We focus on nominals (represented by $\mathcal{O}$), counting quantifiers ($\mathcal{Q}$), self-loops ($\mathcal{S}elf$), inverse roles ($\mathcal{I}$) and the universal role ($\mathcal{U}$). For instance, the logic $\mathcal{ALCUO}$ extends $\mathcal{ALC}$ with the universal role and nominals. Below, we recall the definition of $\mathcal{ALC}$ and the possible additionnal constructors.
 
\begin{definition}[Concept, Role, $\mathcal{ALC}$]
Let $\mathcal{A} = (\mathcal{O}, \mathcal{C}_0, \mathcal{R}_0)$ be an alphabet where $\mathcal{O}$ (resp. $\mathcal{C}_0, \mathcal{R}_0$) is the set of nominals (resp. atomic concepts, atomic roles), given $o \in \mathcal{O}$, $C_0 \in \mathcal{C}_0$, $r_0 \in \mathcal{R}_0$ and $n$ and integer, $\mathcal{ALC}$ concepts $C$ and roles $R$ are defined by:\\
$C := \top 
\; | \; C_0 
\; | \; \exists R. C
\; | \; \neg C  
\; | \; C \vee C$\\

\hspace*{-.5cm}$R := r_0 $

$\mathcal{ALC}$ can be extended by adding some of the following concept and role constructors :\\
$C := o\; (\text{nominals}) \; | \; \exists R. Self\; (\text{self loops}) \; | \; (<\; n\; R\; C)\; (\text{counting quantifiers})$\\
$R := U\; (\text{universal role})\; | \; R^- (\text{inverse role})$

  For the sake of conciseness, we define $\bot \equiv \neg \top$, $C \wedge C' \equiv \neg (\neg C \vee \neg C')$, $\forall R. C \equiv \neg (\exists R. \neg C)$ and $(\geq\; n\; R\; C) \equiv \neg (<\; n\; R\; C)$.
\end{definition}

\begin{definition}[Interpretation]
       An interpretation over an alphabet
    $(\mathcal{C}_0, \mathcal{R}_0, \mathcal{O},)$  is a tuple
    $(\Delta^{\mathcal{I}},\cdot^{\mathcal{I}})$ where
    $\cdot^{\mathcal{I}}$ is a function such that
    $c_0^{\mathcal{I}} \subseteq \Delta^{\mathcal{I}}$, for every
    atomic concept $c_0 \in \mathcal{C}_0$,
    $r_0^{\mathcal{I}} \subseteq \Delta^{\mathcal{I}} \times
    \Delta^{\mathcal{I}}$, for every atomic role $r_0 \in \mathcal{R}_0$,
    $o^\mathcal{I} \in \Delta^{\mathcal{I}}$ for every nominal $o \in
    \mathcal{O}$. The interpretation function is extended to concept
    and role descriptions by the following inductive definitions:  
  \begin{itemize}
  \item $\top^{\mathcal{I}} = \Delta{\mathcal{I}}$
  \item $(\neg C)^{\mathcal{I}} = \Delta{\mathcal{I}} \backslash C^{\mathcal{I}}$
  \item $(C \vee D)^{\mathcal{I}} = C^{\mathcal{I}} \cup
    D^{\mathcal{I}}$
  \item $(\exists R.C)^{\mathcal{I}} = \{n \in \Delta^{\mathcal{I}} |
    \exists m, (n,m) \in R^{\mathcal{I}} \mbox { and } m \in
    C^{\mathcal{I}} \}$
  \item $(\exists R.Self)^{\mathcal{I}} = \{n \in  \Delta^{\mathcal{I}}
    | (n,n) \in R^{\mathcal{I}}\}$
\item $(<\; n\; R\; C)^{\mathcal{I}} = \{\delta \in  \Delta^{\mathcal{I}} | \#(\{m \in \Delta^{\mathcal{I}} | (\delta, m)
   \in R^{\mathcal{I}}$ and $m \in C^{\mathcal{I}}\}) < n\}$
 \item $(R^-)^{\mathcal{I}} = \{(n,m) \in \Delta^{\mathcal{I}} \times
    \Delta^{\mathcal{I}} | (m,n) \in R^\mathcal{I}\}$
  \item $U^\mathcal{I} = \Delta^{\mathcal{I}} \times \Delta^{\mathcal{I}}$
  \end{itemize}
  \end{definition}

\begin{definition}[Interpretation induced by a decorated graph]
  Let $G = (N, E, \phin, \phir, s, t)$ be a graph over an alphabet
  $(\mathcal{C},\mathcal{R})$ such that
  $\mathcal{C}_0 \; \cup\; \mathcal{O} \subseteq \mathcal{C}$ and
  $\mathcal{R}_0 \subseteq \mathcal{R}$. The interpretation induced by
  the graph $G$, denoted $(\Delta^{\mathcal{G}},\cdot^{\mathcal{G}})$
  such that $\Delta^{\mathcal{G}} = N$,
  $c_0^{\mathcal{G}} = \{ n \in N | C_0 \in \phin(n)\}$, for every
  atomic concept $c_0 \in \mathcal{C}_0$,
  $r_0^{\mathcal{G}} = \{(n,m) \in N \times N | \exists e \in E. s(e)
  = n$
  and $t(e) = m$ and $r_0 \in \phir(e)\}$, for every atomic role
  $r_0 \in \mathcal{R}_0$, 
  $o^\mathcal{G} = \{n \in N | o \in \phin(n)\}$ for every nominal
  $o \in \mathcal{O}$.

\noindent
We say that a node $n$ of a graph $G$ satisfies a concept $c$, written
$n \models c$ if $n \in c^{\mathcal{G}}$. We say that a graph $G$
satisfies a concept $c$, written $G \models c$ if $c^{\mathcal{G}} =
N$ that is every node of $G$ belongs to the interpretation of $c$
induced by $G$.
  \end{definition}

We first consider the possibility to express $App(\rho)$ in a
Description Logic $\mathcal{L}$ for a given rule $\rho$. The
definition of $App(\rho)$ depends of the shape of the left-hand side
of $\rho$ on one side and on the expressive power of the considered
logic $\mathcal{L}$. Below, we give a general expression for
$App(\rho)$ for a particular class of left-hand sides and logics
including  $\mathcal{ALCU}$.    

  \begin{proposition}
Let $\mathcal{L}$ be a logic extending $\mathcal{ALCU}$. Let us assume that $\rho$ is a rule whose left-hand side is a tree labeled with $\mathcal{L}$ such that its edges have only one label. $App(\rho)$ can be expressed in $\mathcal{L}$.
  \end{proposition}

\begin{proof}
Let $L = (N^L, E^L, \phin^L, \phir^L, s^L, t^L)$ be the left-hand side of
$\rho$. Let $r$ be the root of $L$. Let $A = \exists U. \psi_n(r) \wedge \bigwedge_{e \in \mathcal{E}(r)^L}  \psi_e(e)$ where $\mathcal{E}(n) = \{e \in E^L | s^L(e) = n\}, \psi_n(n) = \bigwedge_{c \in \phin^L(n)} c$
and\\ $\psi_e(e) = \exists \phir^L(e). \psi_n(t(e))$.

Let us assume that $G = (N^G, E^G, \phin^G, \phir^G, s^G, t^G)$ is
a graph. 

Let us assume that $G \models A$. Then, $\{x_r \in N^G | \exists x_n \in N^G$ for $n \in N^L$, $\exists \xi_e \in e^G$ for $e \in E^L$ such that $\psi_n(n) \in \phin^G(x_n)$, $n = s^L(e) \Rightarrow x_n = s^G(\xi_e)$, $n = t^L(e) \Rightarrow x_n = t^G(\xi_e)$ and $\phir^L(e) \in \phir^G(\xi_e)\}$ is not empty. Then, let us define $h^N(n) = x_n$  and $h^E(e) = \xi_e$.
\begin{enumerate} 
\item For all $n \in N^L$, for all $c \in \phin^L(n)$, $x_n \models c$ by induction.
\item For all $e \in E^L$, for all $r \in \phir^L(e)$, $\xi_{e}
  \models r$
\item  For all $e \in E^L$, $s^G(\xi_e) = x_{s^L(e)}$
\item  For all $e \in E^L$, $t^G(\xi_e) = x_{t^L(e)}$
\end{enumerate}
$(h^N,h^E)$ is thus a match. Hence, there exist at least one.

Let us now assume that there exists a match $(h^N, h^E)$ from $L$ to
$G$. Then, by definition, the $x_n= h^N(n)$'s (and $\xi_e=h^E(e)$'s) defined previously exist. Additionally, due to the first condition, $x_n$ is a model of
$\psi_n$ and, thanks to the other conditions, $\xi_e$ is a model of
$\psi_e$. Thus $G \models A$.

Thus, $A \Leftrightarrow App(\rho)$.   
\end{proof}

\begin{example}
Let us consider the rule $\rho$ of \figref{fig:NAppFO}. $App(\rho)$ can be expressed in $\mathcal{ALCU}$ as $\exists U. (\exists R. (C \wedge \exists R. (C \vee D)))$. 
  \end{example}

We now discuss the closure under substitution of various Description Logics.

\begin{theorem}
The logics $\mathcal{ALCUO}$, $\mathcal{ALCUIO}$,
$\mathcal{ALCUOS}elf$, $\mathcal{ALCUIOS}elf$, $\mathcal{ALCQUIO}$ and
$\mathcal{ALCQUOIS}elf$\footnote{Description Logic names are such that
each letter represents a (groups of) constructor(s). More information can be
found in \cite{DLHandbook}.}
are closed under substitutions.
\end{theorem}

\begin{proof}
  We proved in \cite{BESDL16} that $\mathcal{ALCQUOIS}elf$ is closed
  under substitution for every action but clone and merge. The proof uses a
  rewriting system that replaces formulas with substitutions with
  equivalent formulas without substitutions. It is not possible to
  remove all the substitutions in one step. Some rules are used to move
  the substitutions closer to atomic formulas. 

\begin{figure}
\def\smallscale{0.6}
\begin{center}
\resizebox{5cm}{!}{
\begin{tikzpicture}[->,>=stealth',shorten >=1pt,auto,node distance=2.8cm,
                    semithick,group/.style ={fill=gray!20, node distance=20mm},thickline/.style ={draw, thick, -latex'}]
  \tikzstyle{Cstate}=[circle,fill=white,draw=black,text=black, minimum
  size = 1cm]
  \tikzstyle{Estate}=[circle,fill=none,draw=none,text=black, minimum
  size = 1cm]
  \tikzstyle{NCstate}=[fill=white,draw=black,text=black, minimum
  size = 1cm]
  \tikzstyle{CurrState}=[circle,fill=white,draw=red,text=black, minimum
  size = 1cm]
  \tikzstyle{NCurrState}=[fill=white,draw=red,text=black, minimum
  size = 1cm]

  \node[CurrState] at (0,5) (I1) {\huge $i$};
  \node[NCstate] at (2,5) (J1) {\huge$j$};
  \node[Estate] at (0,5.6) (O1) {};

  \node[Cstate] at (4,5) (I2) {\huge$i$};
  \node[NCstate] at (6,5) (J2) {\huge$j$};
  \node[CurrState] at (5,7) (C2) {};
  
  \node[NCstate] at (-2,0) (I3) {\huge$i$};
  \node[CurrState] at (0,0) (J3) {\huge$j$};
  \node[Estate] at (-2,0.6) (O3) {};

  \node[Cstate] at (2,0) (I4) {\huge$i$};
  \node[NCurrState] at (4,0) (J4) {\huge$j$};
  \node[Estate] at (2,0.6) (O4) {};
  
  \node[Cstate] at (6,0) (I5) {\huge$i$};
  \node[CurrState] at (8,0) (J5) {\huge$j$};
  \node[NCstate] at (7,2) (C5) {};
  
\begin{pgfonlayer}{background}
\node [label=below:{$A$},group, fit=(I1) (J1) (O1)] (a) {};
\node [label=below:{$B$},group, fit=(I2) (J2) (C2)] (b) {};
\node [label=below:{$C_1$},group, fit=(I3) (J3) (O3)] (c) {};
\node [label=below:{$C_2$},group, fit=(I4) (J4) (O4)] (d) {};
\node [label=below:{$C_3$},group, fit=(I5) (J5) (C5)] (e) {};
\end{pgfonlayer}

\path (I1) edge [loop above] node {} (I1)
      (C2) edge node {} (I2)
      (I3) edge [loop above] node {} (I3)
      (I4) edge [loop above] node {} (I4)
      (I5) edge node {} (C5);
      
\path[red,dashed] (I1) edge       node {} (J1)
      (C2) edge node {} (J2)
      (J3) edge       node {} (I3)
      (J4) edge [loop above] node {} (J4)
      (J5) edge node {} (C5);

\end{tikzpicture}
}
\end{center}
\caption{Illustrations of the various ways for a node to satisfy $\exists R. C[cl(i,j,\dots)]$ by gaining a new $R$-neighbor satisfying $C[cl(i,j,\dots)]$. The node where the concept is evaluated is in red; created edges are dashed and red. Squares are nodes that satisfy $C[cl(i,j,\dots)]$. A) $i$ will have a new neighbor $j$ after $[cl(i,j,\dots)]$ if $R \in \rlin$ and it has a self-loop; B) A node that is neither $i$ nor $j$ will have a new neighbor $j$ if $i$ was its neighbor and $ R \in \rin$; C) $j$ will have new neighbours after $[cl(i,j,\dots)]$ if $i$ has a self-loop and $R \in \rlout$ ($C_1$), $i$ has a self-loop and $R \in \rll$ ($C_2$) or if $R \in \rout$ ($C_3$).}
\label{fig:ecl}
\end{figure}
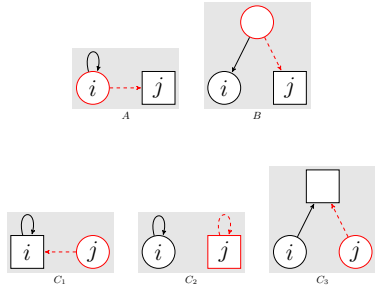

  \begin{itemize}
\item $\top\; \sigma \leadsto \top$
\item $o\; \sigma \leadsto o$
\item $C_0 [add_C(i,C')] \leadsto C_0$
\item $C_0 [del_C(i,C')] \leadsto C_0$
\item $C_0 [add_C(i,C_0)] \leadsto C_0 \vee i$
\item $C_0 [del_C(i,C_0)] \leadsto C_0 \wedge \neg i$
\item $C_0 [add_R(i,j,R)] \leadsto C_0$
\item $C_0 [del_R(i,j,R)] \leadsto C_0$
\item $C_0 [add_N(i)] \leadsto C_0$ if $C_0 \neq Active$
\item $C_0 [del_N(i)] \leadsto C_0 \wedge \neg j$
\item $C_0 [i \gg j] \leadsto C_0$
  \item $C_0[mrg(i,j)] \leadsto \neg j \wedge (C_0 \vee (i \wedge \exists U. (j \wedge C_0)))$ where $C_0$ is an atomic formula different from $Active$
  \item $C_0[cl(i,j,\dots)] \leadsto C_0 \vee (j \wedge \exists U. (i
    \wedge  C_0))$ where $C_0$ is an atomic formula different from
    $Active$
\item $Active [add_N(i)] \leadsto C \vee i$
  \item $Active[mrg(i,j)] \leadsto Active \wedge \neg j$
  \item $Active[cl(i,j,\dots)] \leadsto Active \vee j$
\item $o[\sigma] \leadsto o$
\item $(\neg C)\sigma \leadsto \neg (C\sigma)$
\item $(C \vee D)\sigma \leadsto
  C\sigma \vee
  D\sigma$
\item $\exists R.Self [add_C(i,C_0)] \leadsto \exists
  R.Self$
\item $\exists R.Self [del_C(i,C_0)] \leadsto \exists
  R.Self$
\item $\exists R.Self [add_R(i,j,R')] \leadsto \exists
  R.Self$
\item $\exists R.Self [del_R(i,j,R')] \leadsto \exists
  R.Self$
\item $\exists R.Self [add_R(i,j,R)] \leadsto (\{i\} \wedge
  \{j\}) \vee \exists R.Self$
\item $\exists R.Self [del_R(i,j,R)] \leadsto (\neg \{i\} \vee
 \neg \{j\}) \wedge \exists R.Self$
\item $\exists R.Self [add_N(i)] \leadsto \exists
  R.Self$
\item $\exists R.Self [del_N(i)] \leadsto \exists
  R.Self \wedge \neg \{i\}$
\item $\exists R.Self [i \gg j] \leadsto$ \\$((\{i\} \Leftrightarrow \{j\}) \Rightarrow \exists R.Self) \wedge (\neg \{i\} \wedge \{j\} \Rightarrow \exists R.Self \vee \exists R.\{i\})$
\item $(\exists R.Self)[mrg(i,j)] \leadsto \neg j \wedge (\exists R.Self \vee (\{i\} \wedge (\exists R. \{j\} \vee \exists U. (\{j\} \wedge \exists R.\{i\}) \vee \exists U. (\{j\} \wedge \exists R.Self))))$
\item $(\exists R.Self)[cl(i,j,\dots)] \leadsto \exists R.Self \vee C_S$ where $C_S = \{j\}$ if $R \in \rll$ and $C_S = \bot$ otherwise
\item $(\exists R .\phi)[add_C(i,C_0)] \leadsto \exists R
  .(\phi[add_C(i,C_0)])$
\item $(\exists R .\phi)[del_C(i,C_0)] \leadsto \exists R .(\phi[del_C(i,C_0)])$
\item $(\exists R .\phi)[add_R(i,j,R')] \leadsto \exists R
  .(\phi[add_R(i,j,R')])$
\item $(\exists R .\phi)[del_R(i,j,R')] \leadsto \exists R .(\phi[del_R(i,j,R')])$
\item $(\exists R .\phi)[add_R(i,j,R)] \leadsto (\{i\}
  \wedge$\\ 
\hspace*{-0.7cm}$\exists U. (\{j\}
  \wedge \phi [add_R(i,j,R)])) \vee \exists R. \phi [add_R(i,j,R)]$
\item
$(\exists R .\phi)[del_R(i,j,R)] \leadsto$\\
 \hspace*{-0.7cm}$(\{i\} \Rightarrow \exists R. (\phi [del_R(i,j,R)] \wedge \neg \{j\}))$\\
\hspace*{-0.7cm}$\wedge (\neg \{i\} \Rightarrow \exists R. (\phi
[del_R(i,j,R)]))$
\item $(\exists R .\phi)[add_N(i)] \leadsto \exists R.(\phi[add_N(i)])$\\
\item $(\exists R .\phi)[del_N(i)] \leadsto \neg \{i\} \wedge \exists
  R.(\phi[del_N(i)] \wedge \neg \{i\})$\\
\item $(\exists R .\phi)[i \gg j)] \leadsto (\exists U.(\{i\}
  \wedge \{j\}) \Rightarrow \exists R. \phi [i \gg j])$\\
  \hspace*{5mm}$\wedge (\exists U.(\{i\} \wedge \neg \{j\}) \Rightarrow$\\
\hspace*{1cm}$(\exists R. (\{i\} \wedge \phi[i \gg j]) \wedge \forall R. \neg \{j\} \wedge \exists U.(\{j\} \wedge \neg \phi[i \gg j])
\Rightarrow$\\
\hspace*{1.5cm}$\exists R. (\phi[i \gg j] \wedge \neg \{i\}))$\\
\hspace*{1cm}$\wedge
(\exists R. \{i\} \wedge \forall R. \neg \{j\} \wedge \exists U.(\{j\} \wedge \phi[i \gg j]))$\\
\hspace*{1cm}$\wedge
(\exists R. (\{i\} \wedge \phi[i \gg j]) \wedge \exists R. \{j\}
\Rightarrow$\\
\hspace*{1.5cm}$\exists R. (\phi[i \gg j] \wedge \neg \{i\}))$\\
\hspace*{1cm}$\wedge
((\forall R. \neg \{i\})$\\
\hspace*{1.2cm}$\vee (\exists R. (\{i\} \wedge \neg \phi[i \gg j]) \wedge \exists R. \{j\})$\\
\hspace*{1.2cm}$\vee (\exists R.(\{i\} \wedge \neg \phi[i \gg j]) \wedge \forall R. \neg \{j\} \wedge \exists U.(\{j\} \wedge \neg \phi[i \gg j]))
\Rightarrow$\\
\hspace*{1.5cm}$\exists R. \phi[i \gg j]))$
  \item $(\exists R.\phi)[mrg(i,j)] \leadsto \neg \{j\} \wedge (\exists R. (\neg
    \{j\} \wedge \phi[mrg(i,j)]) \vee \exists R. (\neg
    \{j\} \wedge \phi[mrg(i,j)]) \vee \{i\} \wedge
\exists U. ((\{i\} \vee \{j\}) \wedge \exists R. \phi[mrg(i,j)]))$
  \item $(\exists R.\phi)[cl(i,j,\dots)] \leadsto$\\
  $\exists R. (\phi[cl(i,j,\dots)]) \vee \cin \vee \cout \vee \clin \vee \clout \vee \cll$ where:
    \begin{itemize}
  \item $\cin = \left\{
	\begin{array}{ll}
	  \begin{array}{l}
	  \neg \{i\} \wedge \neg \{j\} \wedge \exists R. \{i\} \wedge\\
	  (\exists U. (\{j\} \wedge \phi[cl(i,j,\dots)]))
	  \end{array} & \mbox{if } R \in \rin\\
          \bot  & \mbox{otherwise }
	\end{array}
        \right.$
  \item $\cout = \left\{
	\begin{array}{ll}
	  \{j\} \wedge (\exists U.(\{i\} \wedge \exists R.( \neg \{i\} \wedge \phi[cl(i,j,\dots)]))) & \mbox{if } R \in \rout\\
          \bot  & \mbox{otherwise }
	\end{array}\right.$
  \item $\clin = \left\{
	\begin{array}{ll}
	   \{i\} \wedge \exists R. \{i\} \wedge \exists U.(\{j\} \wedge \phi[cl(i,j,\dots)])  & \mbox{if } R \in \rlin\\
          \bot  & \mbox{otherwise }
	\end{array}
        \right.$
  \item $\clout = \left\{
	\begin{array}{ll}
	    \{j\} \wedge \exists U. (\{i\} \wedge \exists R.\{i\} \wedge \phi[cl(i,j,\dots)]) & \mbox{if } R \in \rlout\\
          \bot  & \mbox{otherwise }
	\end{array}
        \right.$
  \item $\cll = \left\{
	\begin{array}{ll}
	  \{j\} \wedge \phi[cl(i,j,\dots)] \wedge \exists U. (\{i\} \wedge \exists R.\{i\})  & \mbox{if } R \in \rll\\
          \bot  & \mbox{otherwise }
	\end{array}
        \right.$
    \end{itemize}
\item$(<\; n\; R\; \phi) [add_C(i,C_0)] \leadsto
  (<\; n\; R\; \phi [add_C(i,C_0)]) $
\item$(<\; n\; R\; \phi) [del_C(i,C_0)] \leadsto
  (<\; n\; R\; \phi [del_C(i,C_0)]) $
\item $(<\; n\; R\; \phi) [add_R(i,j,R')] \leadsto (<\;
  n\; R\; \phi [add_R(i,j,R')]) $
\item $(<\; n\; R\; \phi) [del_R(i,j,R')] \leadsto (<\; n\; R\; \phi [del_R(i,j,R')]) $
\item
$(<\; n\; R\; \phi) [add_R(i,j,R)] \leadsto$
\vspace{-0.2cm}\begin{tabbing}
\hspace*{-0.5cm} \= $((\{i\} \wedge \exists U. (\{j\}\; \wedge \phi [add_R(i,j,R)])  \wedge \forall R. \neg \{j\})$ \hspace{0.3cm}\= $\Rightarrow$ \\
  \>  $(<\; (n-1)\; R\; \phi [add_R(i,j,R)]) )$\\
  \hspace*{-0.7cm}$\wedge$ \> $((\neg \{i\} \vee \forall U. (\neg \{j\}
   \vee \neg \phi [add_R(i,j,R)]) \vee \exists R. \{j\})$ \>
   $\Rightarrow$\\
  \> $(<\; n\; R\; \phi [add_R(i,j,R)]))$
\end{tabbing}
\item
$(<\; n\; R\; \phi) [del_R(i,j,R)] \leadsto$
\vspace{-0.2cm}\begin{tabbing}
\hspace*{-0.5cm} 
\= $((\{i\} \wedge \exists U. (\{j\}\; \wedge \phi [del_R(i,j,R)])
 \wedge \exists R. \{j\})$ \hspace{0.8cm}\= $\Rightarrow$ \\
 \>  $(<\; (n+1)\; R\; \phi [del_R(i,j,R)]) )$\\
 \hspace*{-0.7cm}$\wedge$ \> $((\neg \{i\} \vee \forall U. (\neg \{j\}
  \vee \neg \phi [del_R(i,j,R)]) \vee \forall R. \neg \{j\})$ \>
  $\Rightarrow$\\
 \> $(<\; n\; R\; \phi [del_R(i,j,R)]))$
\end{tabbing}
\item $(<\; n\; R\; \phi) [add_N(i)] \leadsto (<\; n\; R\;
  \phi)$
\item $(<\; n\; R\; \phi) [del_N(i)] \leadsto \{i\} \vee (<\; n\; R\;
  (\phi[del_N(i)] \wedge \neg \{i\}))$
\item $(<\; n\; R\; \phi) [i \gg j] \leadsto (\exists
  U.(\{i\} \wedge \{j\}) \Rightarrow (<\; n\; R\; \phi [i
  \gg j]))$\\
$\wedge (\exists U.(\{i\} \wedge \neg \{j\})
  \Rightarrow$\\
\hspace*{5mm}$(\exists R. (\{i\} \wedge \phi[i \gg j]) \wedge \forall R. \neg \{j\} \wedge \exists U.(\{j\} \wedge \neg \phi[i \gg j])
\Rightarrow$\\
\hspace*{1cm}$ (<\; (n + 1)\; R\; \phi[i \gg j]))$\\
\hspace*{5mm}$ \wedge
(\exists R. (\{i\} \wedge \neg \phi[i \gg j]) \wedge \forall R. \neg \{j\} \wedge \exists U.(\{j\} \wedge \phi[i \gg j])
\Rightarrow$\\
\hspace*{1cm}$ (<\; (n - 1)\; R\; \phi[i \gg j]))$\\
\hspace*{5mm}$ \wedge
(\exists R. (\{i\} \wedge \phi[i \gg j]) \wedge \exists R. \{j\}
\Rightarrow$\\
\hspace*{1cm}$ (<\; (n + 1)\; R\; \phi[i \gg j]))$\\
\hspace*{5mm}$\wedge
((\forall R. \neg \{i\})$\\
 \hspace*{1cm} $\vee (\exists R. (\{i\} \wedge \neg \phi[i \gg
 j]) \wedge \exists R. \{j\})$\\
\hspace*{1cm} $\vee (\exists R. (\{i\} \wedge \phi[i \gg j])
\wedge \forall R. \neg \{j\} \wedge \exists U.(\{j\} \wedge \phi[i
\gg j]))$\\
\hspace*{1cm}$\vee (\exists R.(\{i\} \wedge \neg \phi[i \gg j]) \wedge \forall R. \neg \{j\} \wedge \exists U.(\{j\} \wedge \neg \phi[i \gg j]))
\Rightarrow $\\
\hspace*{1.5cm}$(<\; n\; R\; \phi[i \gg j])))$
\item $(<\; n\; R^-\; \phi) [i \gg j] \leadsto (\{i\} \wedge
  \neg \{j\}) \vee$\\
 \hspace*{5mm}$(\neg \{i\} \wedge \{j\} \Rightarrow$\\
\hspace*{1cm}$\bigsqcup_{k \in
   [0,n]}(<\; k\; R^-\; \phi [i \gg j]) \wedge$\\
\hspace*{1.5cm}$\exists U. (\{i\}
 \wedge (<\; (n-k)\; R^-\; (\phi [i \gg j] \wedge \neg \exists
 R^-.\{j\}))))$\\
\hspace*{5mm}$\vee
((\{i\} \Leftrightarrow \{j\}) \Rightarrow (<\; n\; R^-\; \phi [i \gg
j]))$
    \item $(<\; n\; R\; \phi)[mrg(i,j)] \leadsto \{j\} \vee \\(\{i\}
      \wedge \bigvee_{k = 1}^{n} (< \; k\; R\; (\phi[mrg(i,j)] \wedge
      \forall R^-. \neg \{j\})) \wedge \exists U. (\{j\} \wedge (<\;
      (n - k)\; R\; \phi[mrg(i,j)] ))) \vee (\neg\{i\} \wedge
      \neg\{j\} \wedge \\
(\exists R.(\{j\} \wedge \neg \phi[mrg(i,j)]) \wedge \forall
R.\neg\{i\} \wedge \exists U.(\{i\} \wedge \phi[mrg(i,j)]) \wedge (<\;
n - 1\; R\; \phi[mrg(i,j)])) \vee\\
(\exists R.(\{i\} \wedge \neg \phi[mrg(i,j)]) \wedge \forall
R.\neg\{j\} \wedge \exists U.(\{j\} \wedge \phi[mrg(i,j)]) \wedge (<\;
n - 1\; R\; \phi[mrg(i,j)])) \vee\\
(\exists R.(\{i\} \wedge \phi[mrg(i,j)]) \wedge \exists R.(\{j\} \wedge \phi[mrg(i,j)]) \wedge (<\;
n + 1\; R\; \phi[mrg(i,j)])) \vee\\
 ((\forall R.(\neg \{j\} \vee \phi[mrg(i,j)]) \vee \exists
R.(\{i\} \wedge \neg \phi[mrg(i,j)])) \wedge (\forall R.(\neg \{i\} \vee \phi[mrg(i,j)]) \vee \exists
R.(\{j\} \wedge \neg \phi[mrg(i,j)])) \wedge (\forall R.(\neg \{i\}
\vee \neg\phi[mrg(i,j)]) \vee \forall R.(\neg \{j\} \vee
\neg\phi[mrg(i,j)])) \wedge (<\; n\; R\; \phi[mrg(i,j)])))$
    \item $(<\; n\; R\; \phi)[cl(i,j,\dots)] \leadsto (\{i\} \Rightarrow C_i) \wedge (\{j\} \Rightarrow C_j) \wedge (\neg \{i\} \wedge \neg \{j\} \Rightarrow C_o)$ where:
    \begin{itemize}
    \item $C_i = (<\; n\; R\; \phi[cl(i,j,\dots)])$ if $R \not\in \rlin$, and
    \item $C_i = (\exists R. \{i\} \wedge \exists U. (\{j\} \wedge \phi[cl(i,j,\dots)]) \Rightarrow\\$
      \hspace*{1cm}$(<\; n-1\; R\; \phi[cl(i,j,\dots)])) \wedge$\\
      $(\forall R. \neg \{i\} \vee \exists U. (\{j\} \wedge \neg \phi[cl(i,j,\dots)]) \Rightarrow\\$
      \hspace*{1cm}$(<\; n\; R\; \phi[cl(i,j,\dots)]))$ if $R \in \rlin$
    \item $C_j = \top$ if $R \not\in \rout$ and either:
      \begin{itemize}
      \item $R \not\in \rlout \cup \rll$, or
      \item $R \not\in \rlout \cap \rll$ and $n > 1$, or
      \item $n > 2$
    \end{itemize}, and
    \item $C_j = (\exists U. (\{i\} \wedge \exists R. \{i\} \wedge \phi[cl(i,j,\dots)]) \Rightarrow$\\
      \hspace*{1cm}$\bot) \wedge$\\
      $(\exists U. (\{i\} \wedge (\forall R. \neg \{i\} \vee \neg \phi[cl(i,j,\dots)])) \Rightarrow$\\
      \hspace*{1cm}$\top)$ if $R \not\in \rout \cup \rll$ and $R \in \rlout$ and $n = 1$, and
    \item $C_j = (\exists U. (\{i\} \wedge \exists R. \{i\}) \wedge \phi[cl(i,j,\dots)] \Rightarrow$\\
      \hspace*{1cm}$\bot) \wedge$\\
      $(\exists U. (\{i\} \wedge \forall R. \neg \{i\}) \vee \neg \phi[cl(i,j,\dots)])) \Rightarrow$\\
      \hspace*{1cm}$\top)$ if $R \not\in \rout \cup \rlout$ and $R \in \rll$ and $n = 1$, and
    \item $C_j = (\exists U. (\{i\} \wedge \exists R. \{i\} \wedge \phi[cl(i,j,\dots)]) \wedge \phi[cl(i,j,\dots)] \Rightarrow$\\
      \hspace*{1cm}$\bot) \wedge$\\
      $(\exists U. (\{i\} \wedge (\forall R. \neg \{i\} \vee \neg \phi[cl(i,j,\dots)])) \vee \neg \phi[cl(i,j,\dots)] \Rightarrow$\\
      \hspace*{1cm}$\top)$ if $R \not\in \rout$ and $R \in \rlout \cap \rll$ and $n = 2$, and
    \item $C_j = (\exists U. (\{i\} \wedge \exists R. \{i\}) \wedge (\phi[cl(i,j,\dots)]) \vee \exists U. (\{i\} \wedge \phi[cl(i,j,\dots)]) \Rightarrow$\\
      \hspace*{1cm}$\bot) \wedge$\\
      $(\exists U. (\{i\} \wedge \forall R. \neg \{i\}) \vee (\neg \phi[cl(i,j,\dots)] \wedge \exists U. (\{i\} \wedge \neg \phi[cl(i,j,\dots)])) \Rightarrow$\\
      \hspace*{1cm}$\top)$ if $R \not\in \rout$ and $R \in \rlout \cap \rll$ and $n = 1$, and
    \item $C_j = (\exists U. (\{i\} \wedge \exists R. \{i\} \wedge \phi[cl(i,j,\dots)]) \Rightarrow$\\
      \hspace*{1cm}$\exists U. (\{i\} \wedge (<\; n-1\; R\; (\neg i \wedge \phi[cl(i,j,\dots)])))) \wedge$\\
      $(\exists U. (\{i\} \wedge (\forall R. \neg \{i\} \vee \neg \phi[cl(i,j,\dots)])) \Rightarrow$\\
      \hspace*{1cm}$\exists U. (\{i\} \wedge (<\; n\; R\; (\neg \{i\} \wedge \phi[cl(i,j,\dots)]))))$ if $R \in \rout \cup \rlout$ and $R \not\in \rll$, and
    \item $C_j = (\exists U. (\{i\} \wedge \exists R. \{i\}) \wedge \phi[cl(i,j,\dots)] \Rightarrow$\\
      \hspace*{1cm}$\exists U. (\{i\} \wedge (<\; n-1\; R\; (\neg \{i\} \wedge \phi[cl(i,j,\dots)])))) \wedge$\\
      $(\exists U. (\{i\} \wedge \forall R. \neg \{i\}) \vee \neg \phi[cl(i,j,\dots)] \Rightarrow$\\
      \hspace*{1cm}$\exists U. (\{i\} \wedge (<\; n\; R\; (\neg \{i\} \wedge \phi[cl(i,j,\dots)]))))$ if $R \in \rout \cup \rll$ and $R \not\in \rlout$, and
    \item $C_j = (\exists U. (\{i\} \wedge \exists R. \{i\} \wedge \phi[cl(i,j,\dots)]) \wedge \phi[cl(i,j,\dots)] \Rightarrow$\\
      \hspace*{1cm}$\exists U. (\{i\} \wedge (<\; n-2\; R\; (\neg \{i\} \wedge \phi[cl(i,j,\dots)])))) \wedge$\\
      $(\exists U. (\{i\} \wedge \exists R. \{i\} \wedge \neg \phi[cl(i,j,\dots)]) \wedge \phi[cl(i,j,\dots)] \Rightarrow$\\
      \hspace*{1cm}$\exists U. (\{i\} \wedge (<\; n-1\; R\; (\neg \{i\} \wedge \phi[cl(i,j,\dots)])))) \wedge$\\
      $(\exists U. (i \wedge \exists R. \{i\} \wedge \phi[cl(i,j,\dots)]) \wedge \neg \phi[cl(i,j,\dots)] \Rightarrow$\\
      \hspace*{1cm}$\exists U. (\{i\} \wedge (<\; n-1\; R\; (\neg \{i\} \wedge \phi[cl(i,j,\dots)])))) \wedge$\\
      $(\exists U. (\{i\} \wedge \forall R. \neg \{i\}) \wedge \neg \phi[cl(i,j,\dots)]) \vee \exists U. (\{i\} \wedge \neg \phi[cl(i,j,\dots)]) \Rightarrow$\\
      \hspace*{1cm}$\exists U. (\{i\} \wedge (<\; n\; R\; (\neg \{i\} \wedge \phi[cl(i,j,\dots)]))))$ if $R \in \rout \cap \rlout \cap \rll$      
    \item $C_o = (<\; n\; R\; \phi[cl(i,j,\dots)])$ if $R \not\in \rin$, and
    \item $C_o = (\exists R. \{i\} \wedge \exists U. (\{j\} \wedge \phi[cl(i,j,\dots)]) \Rightarrow\\$
      \hspace*{1cm}$(<\; n-1\; R\; \phi[cl(i,j,\dots)])) \wedge$\\
      $(\forall R. \neg \{i\} \vee \exists U. (\{j\} \wedge \neg \phi[cl(i,j,\dots)]) \Rightarrow\\$
      \hspace*{1cm}$(<\; n\; R\; \phi[cl(i,j,\dots)]))$ if $R \in \rin$  
    \end{itemize}
\end{itemize}

 We gave an illustration of the various possible cases for $(\exists R.C)[cl(i,j,\dots)]$ in \figref{fig:ecl}. As illustrated by the equivalence given, there are 2 ways for a node to satisfy $(\exists R.C)[cl(i,j,\dots)]$: either it already had such a neighbor before cloning or it gained it during cloning. The 5 possible ways for the second scenario to happen are given in \figref{fig:ecl}. 
 Using this picture, one can also see the various cases of $(<\; n\;
 R\; C)[cl(i,j,\dots)]$. As it is quite complex and depends on $\rin$,
 $\rout$, $\rlin$, $\rlout$, $\rll$ and $n$, we do not report the
 exact equivalence. We give an idea of what it is, though: assuming
 $j$ will be labeled with $C$, we remark that in case A) $i$ needs $n
 - 1$ neighbors that will be labeled with $C$ and it needs $n$
 otherwise, and, in case B), the same can be said for other nodes. $j$
 is more problematic. If $R \not\in \rout$, it will have at most 2
 neighbors, if $R \in \rout$ it will have as many as $i$ plus,
 possibly, $i$ and $j$.

We give an illustration of the counting quantifiers in the case of $mrg(i,j)$ in \figref{fig:CQmrg}. $j$ always satisfies $(<\; n\; R\; C)[mrg(i,j)]$ as it has no neighbors after merging. $i$ has its neighbors plus those of $j$ that were not already its neighbors. All other nodes can either gain a new one $i$, lose one $j$ or both.

\begin{figure}
\def\smallscale{0.6}
\begin{center}
\resizebox{6cm}{!}{
\begin{tikzpicture}[->,>=stealth',shorten >=1pt,auto,node distance=2.8cm,
                    semithick,group/.style ={fill=gray!20, node distance=20mm},thickline/.style ={draw, thick, -latex'}]
  \tikzstyle{Cstate}=[circle,fill=white,draw=black,text=black, minimum
  size = 1cm]
    \tikzstyle{NCstate}=[fill=white,draw=black,text=black, minimum
  size = 1cm]
  \tikzstyle{Astate}=[circle,fill=none,draw=none,text=black, minimum
  size = 1cm]
  \tikzstyle{CurrState}=[circle,fill=white,draw=red,text=black, minimum
  size = 1cm]
  
  \node[Cstate] at (0,10) (I1) {\huge $i$};
  \node[CurrState] at (2,10) (J1) {\huge$j$};
  \node[Astate] at (0.5,11) (O1) {};

  \node[CurrState] at (4,10) (I2) {\huge$i$};
  \node[Cstate] at (6,10) (J2) {\huge$j$};
  \node[NCstate] at (4,12) (C2) {};
  \node[NCstate] at (6,12) (D2) {};
 
  \node[NCstate] at (-2,5) (I3) {\huge$i$};
  \node[Cstate] at (0,5) (J3) {\huge$j$};
  \node[CurrState] at (-1,7) (C3) {};

  \node[NCstate] at (2,5) (I4) {\huge$i$};
  \node[NCstate] at (4,5) (J4) {\huge$j$};
  \node[CurrState] at (3,7) (C4) {};
  
  \node[Cstate] at (6,5) (I5) {\huge$i$};
  \node[NCstate] at (8,5) (J5) {\huge$j$};
  \node[CurrState] at (7,7) (C5) {};
  
  \node[Cstate] at (-2,0) (I6) {\huge$i$};
  \node[Cstate] at (0,0) (J6) {\huge$j$};
  \node[CurrState] at (-1,2) (C6) {};

  \node[Cstate] at (2,0) (I7) {\huge$i$};
  \node[Cstate] at (4,0) (J7) {\huge$j$};
  \node[CurrState] at (3,2) (C7) {};
  
  \node[NCstate] at (6,0) (I8) {\huge$i$};
  \node[NCstate] at (8,0) (J8) {\huge$j$};
  \node[CurrState] at (7,2) (C8) {};
  
\begin{pgfonlayer}{background}
\node [label=below:{$A$},group, fit=(I1) (J1) (O1)] (a) {};
\node [label=below:{$B$},group, fit=(I2) (J2) (C2) (D2)] (b) {};
\node [label=below:{$C_1$},group, fit=(I3) (J3) (C3)] (c) {};
\node [label=below:{$C_2$},group, fit=(I4) (J4) (C4)] (d) {};
\node [label=below:{$C_3$},group, fit=(I5) (J5) (C5)] (e) {};
\node [label=below:{$C_4$},group, fit=(I6) (J6) (C6)] (f) {};
\node [label=below:{$C_5$},group, fit=(I7) (J7) (C7)] (g) {};
\node [label=below:{$C_6$},group, fit=(I8) (J8) (C8)] (h) {};
\end{pgfonlayer}

\path (I2) edge node {} (C2)
      (J2) edge node {} (C2)
      (J2) edge node {} (D2)
      (C3) edge node {} (J3)
      (C4) edge  node {} (I4)
      (C4) edge  node {} (J4)
      (C5) edge  node {} (J5)
      (C7) edge  node {} (J7)
      (C8) edge  node {} (J8);

\end{tikzpicture}
}
\end{center}
\caption{Illustrations of the various ways for a node to satisfy $(<\; n\; R C)[mrg(i,j)]$ when the merging action affects the number of neighbors of a node. The node where the concept is evaluated is in red, the nodes that will be labeled with $C$ are squares. A) $j$ has no remaining neighbor, it thus satisfies $(<\; n\; R C)[mrg(i,j)]$. B) $i$ will have as neighbors all its neighbors plus those of $j$. It is important to count each one only once. C) If the node is neither $i$ nor $j$, it will gain a new neighbor that will be labeled with $C$ - $i$ - if $i$ will be labeled with $C$, it is not yet a neighbor and $j$ is a neighbor that would not be labeled with $C$ ($C_1$); on the other hand, it will lose a neighbor that will be labeled with $C$ - $j$ - if $j$ is a neighbor that will be labeled with $C$ and either $i$ is also a neighbor that will be labeled with $C$ ($C_2$) or $i$ will not be labeled with $C$ ($C_3$); otherwise, the number of neighbors that will be labeled with $C$ stays the same either because there is no new neighbor ($C_4$), because neither $i$ nor $j$ will be labeled with $C$ ($C_5$) or because it loses one neighbor that will be labeled with $C$ - $j$ - and gains one - $i$ ($C_6$).}
\label{fig:CQmrg}
\end{figure}
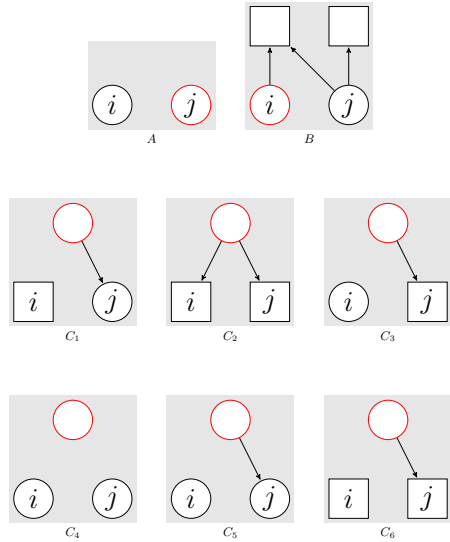

Let us now prove than the two sides of each rules are indeed
equivalent.

\begin{description}
  \item[{$\top\sigma$:}] By definition, $\top$ is always satisfied
  \item[{$o\sigma$:}] The interpretation of $o$ is never modified
  \item[{$C_0[add_C(i,C')]$:}] The interpretation of $C_0$ does not depend on the
interpretation $C'$.
  \item[{$C_0[del_C(i,C')]$:}] The interpretation of $C_0$ does not depend on the
interpretation $C'$.
  \item[{$C_0[add_C(i,C_0)]$:}] The interpretation of $C_0$ becomes
    $C_0^\mathcal{I} \cup i^\mathcal{I} = (C_0 \vee i)^\mathcal{I}$.
  \item[{$C_0[del_C(i,C_0)]$:}] The interpretation of $C_0$ becomes
    $C_0^\mathcal{I} \backslash i^\mathcal{I} = (C_0 \wedge \neg
    i)^\mathcal{I}$.
  \item[{$C_0[add_R(i,j,R)]$:}] The interpretation of $C_0$ does not depend on the
interpretation $R$.
  \item[{$C_0[del_R(i,j,R)]$:}] The interpretation of $C_0$ does not depend on the
interpretation $R$.
  \item[{$C_0[add_N(i)]$:}] The interpretation of $C_0$ does not depend on the
interpretation $R$.
  \item[{$C_0[del_N(i)]$:}] The interpretation of $C_0$ becomes
    $C_0^\mathcal{I} \backslash i^\mathcal{I} = (C_0 \wedge \neg
    i)^\mathcal{I}$.
  \item[{$C_0[i \gg j]$:}] The interpretation of $C_0$ does not depend on the
interpretation of any role.
  \item[{$C_0[mrg(i,j)]$:}] The interpretation of $C_0$  becomes
    $C_0^\mathcal{I} \cup \{n | n = i^\mathcal{I} \wedge j \in C_0^\mathcal{I}\} \backslash j^\mathcal{I} = (\neg
    j \wedge (C_0 \vee (i \wedge \exists U.(j \wedge C_0))))^\mathcal{I}$.
  \item[{$C_0[cl(i,j, \dots)]$:}] The interpretation of $C_0$  becomes
    $C_0^\mathcal{I} \cup \{n | n = j^\mathcal{I} \wedge i \in C_0^\mathcal{I}\} \backslash j^\mathcal{I} = (C_0 \vee (j \wedge \exists U.(i \wedge C_0)))^\mathcal{I}$.
  \item[{$Active[add_N(i)]$:}] The interpretation of $Active$ becomes
    $N^G \cup i^\mathcal{I} = (Active \vee i)^\mathcal{I}$
  \item[{$Active[mrg(i,j)]$:}] The interpretation of $Active$ becomes
    $N^G \backslash j^\mathcal{I} = (Active \wedge \neg j)^\mathcal{I}$
  \item[{$Active[cl(i,j)]$:}] The interpretation of $Active$ becomes
    $N^G \cup j^\mathcal{I} = (Active \vee j)^\mathcal{I}$ 
  \item[{$o\sigma$:}] The interpretation of $o$ is never modified
  \item[{$(\neg C)\sigma$:}] As $(\neg C)^\mathcal{I'} =
    \neg (C^\mathcal{I'})$, $(\neg C)^\mathcal{I'} =
    \neg (C\sigma)^\mathcal{I}$
  \item[{$(C \vee D)\sigma$:}] As $(C \vee D)^\mathcal{I'} =
    C^\mathcal{I'} \vee D^\mathcal{I'}$, $(C \vee D)^\mathcal{I'} =
    (C\sigma \vee D\sigma)^\mathcal{I}$
  \item[{$\exists R. Self[add_C(i,C_0)]$:}] The interpretation of $R$
    does not depend on the interpretation of $C_0$.
  \item[{$\exists R. Self[del_C(i,C_0)]$:}] The interpretation of $R$
    does not depend on the interpretation of $C_0$.
  \item[{$\exists R. Self[add_R(i,j,R')]$:}] The interpretation of $R$
    does not depend on the interpretation of $R'$.
  \item[{$\exists R. Self[del_R(i,j,R')]$:}] The interpretation of $R$
    does not depend on the interpretation of $R'$.
  \item[{$\exists R. Self[add_R(i,j,R)]$:}] The interpretation of $R$
    becomes $R^\mathcal{I} \cup i^\mathcal{I} \times
    j^\mathcal{I}$. Thus $(\exists R. Self)^\mathcal{I'} = \{ n \in
    \Delta | \exists e \in R^\mathcal{I}. s(e) = n$ and $t(e) =n\}
    \cup \{n \in \Delta | n = i^\mathcal{I}$ and $n = j^\mathcal{I}\}$
      that is $(\exists R. Self)^\mathcal{I'} = (\exists
      R. Self \vee (i \wedge j))^\mathcal{I} $
  \item[{$\exists R. Self[del_R(i,j,R)]$:}] The interpretation of $R$
    becomes $R^\mathcal{I} \backslash i^\mathcal{I} \times
    j^\mathcal{I}$. Thus $(\exists R. Self)^\mathcal{I'} = \{ n \in
    \Delta | (n,n) \in R^\mathcal{I}\}
    \backslash \{n \in \Delta | n = i^\mathcal{I}$ and $n = j^\mathcal{I}\}$
      that is $(\exists R. Self)^\mathcal{I'} = (\exists
      R. Self \wedge (\neg i \vee \neg j))^\mathcal{I} $
  \item[{$\exists R. Self[add_N(i)]$:}] The interpretation of $R$
    is not modified.
  \item[{$\exists R. Self[del_N(i)]$:}] The interpretation of $R$
    becomes $R^\mathcal{I} \backslash \{e | s^\mathcal{I}(e) = i^\mathcal{I}$ or
    $t^\mathcal{I}(e) = i^\mathcal{I}\}$. Thus $(\exists R. Self)^\mathcal{I'} = (\exists
      R. Self \wedge \neg i)^\mathcal{I} $
  \item[{$\exists R. Self[i \gg j]$:}] Let us assume that $(n,n) \in
    R^\mathcal{I'}$. Then, either:
\begin{itemize}
\item  $n \neq i^\mathcal{I}$ and $n \neq j^\mathcal{I}$ and thus
  $(n,n) \in R^\mathcal{I}$
\item or $n = i^\mathcal{I} = j^\mathcal{I}$ and thus
  $(n,n) \in R^\mathcal{I}$
\item or $n = j^\mathcal{I} \neq i^\mathcal{I}$ and thus either
  $(j^\mathcal{I},j^\mathcal{I}) \in R^\mathcal{I}$ or
  $(j^\mathcal{I},i^\mathcal{I}) \in R^\mathcal{I}$
\item or $n = i^\mathcal{I} \neq j^\mathcal{I}$ which is impossible as
  $R^\mathcal{I} \cap \{(n,i^\mathcal{I})\} = \emptyset$
\end{itemize}
Thus $(\exists R.Self[i \gg j])^\mathcal{I} = (((\{i\} \Leftrightarrow \{j\}) \Rightarrow \exists R.Self) \wedge (\neg \{i\} \wedge \{j\} \Rightarrow \exists R.Self \vee \exists R.\{i\}))^\mathcal{I}$
  \item[{$\exists R. Self[mrg(i,j)]$:}] Let us assume that $(n,n) \in
    R^\mathcal{I'}$. Then, either:
\begin{itemize}
\item  $n = j^\mathcal{I}$ which is impossible
\item or $n = i^\mathcal{I}$ and thus one of
  $(i^\mathcal{I},i^\mathcal{I})$, $(i^\mathcal{I},j^\mathcal{I})$,
  $(j^\mathcal{I}, i^\mathcal{I})$ or $(j^\mathcal{I},j^\mathcal{I})
  \in R^\mathcal{I}$
\item or $n \neq j^\mathcal{I}$ and $n \neq i^\mathcal{I}$ and thus
  $(n,n) \in R^\mathcal{I}$
\end{itemize}
Thus $(\exists R.Self[mrg(i,j)])^\mathcal{I} = (\neg \{j\} \wedge (\exists
R.Self \vee (\{i\} \wedge (\exists R.\{j\} \vee \exists U.(\{j\}
\wedge \exists R.\{i\}) \vee \exists U. (\{j\} \wedge \exists R.Self)))))^\mathcal{I}$
  \item[{$\exists R. Self[cl(i,j,\dots)]$:}] Let us assume that $(n,n) \in
    R^\mathcal{I'}$. Then, either:
\begin{itemize}
\item  $n = j^\mathcal{I}$ and $R \in \rll$
\item or $n \neq j^\mathcal{I}$ and thus $(n,n) \in R^\mathcal{I}$
\end{itemize}
Thus $(\exists R.Self[mrg(i,j)])^\mathcal{I} = (\exists R. Self \vee
C_S)^\mathcal{I}$ where $C_S = \{j\}$ if $R \in \rll$ and $\bot$ otherwise.
\item[{$(\exists R. \phi)[add_C(i,C_0)]$:}] As the valuation of $R$ is not
  modified by the substitution, $((\exists
  R. \phi)[add_C(i,C_0)]))^\mathcal{I} = (\exists
  R. (\phi[add_C(i,C_0)]))^\mathcal{I}$.
\item[{$(\exists R. \phi)[del_C(i,C_0)]$:}] As the valuation of $R$ is not
  modified by the substitution, $((\exists
  R. \phi)[del_C(i,C_0)]))^\mathcal{I} = (\exists
  R. (\phi[del_C(i,C_0)]))^\mathcal{I}$.
\item[{$(\exists R. \phi)[add_R(i,j,R')]$:}] As the valuation of $R$ is not
  modified by the substitution, $((\exists
  R. \phi)[add_R(i,j,R')]))^\mathcal{I} = (\exists
  R. (\phi[add_R(i,j,R')]))^\mathcal{I}$.
\item[{$(\exists R. \phi)[del_R(i,j,R')]$:}] As the valuation of $R$ is not
  modified by the substitution, $((\exists
  R. \phi)[del_R(i,j,R')]))^\mathcal{I} = (\exists R. (\phi[del_R(i,j,R')]))^\mathcal{I}$.
\item[{$(\exists R. \phi)[add_R(i,j,R)]$:}] As the valuation of $R$
  becomes $R^\mathcal{I} \cup \{(i^\mathcal{I},j^\mathcal{I})\}$, $((\exists
  R. \phi)[add_R(i,j,R)]))^\mathcal{I} = (\exists
  R. (\phi[add_R(i,j,R')]) \vee (\{i\} \wedge \exists U.(\{j\} \wedge \phi[add_R(i,j,R)])))^\mathcal{I}$.
\item[{$(\exists R. \phi)[del_R(i,j,R)]$:}] As the valuation of $R$
  becomes $R^\mathcal{I} \backslash \{(i^\mathcal{I},j^\mathcal{I})\}$, $((\exists
  R. \phi)[del_R(i,j,R)]))^\mathcal{I} = (\exists
  R. (\phi[del_R(i,j,R')]) \wedge (\neg \{i\} \vee \exists R.(\neg
  \{j\} \wedge \phi[del_R(i,j,R)])))^\mathcal{I}$.
\item[{$(\exists R. \phi)[add_N(i)]$:}] As the valuation of $R$ is not
  modified by the substitution, $((\exists
  R. \phi)[add_N(i)]))^\mathcal{I} = (\exists
  R. (\phi[add_N(i)]))^\mathcal{I}$.
\item[{$(\exists R. \phi)[del_N(i)]$:}] As the valuation of $R$
  becomes $R^\mathcal{I} \backslash \{(n,n') | n = i^\mathcal{I}$ or
  $ n' = i^\mathcal{I}\}$, $((\exists
  R. \phi)[del_N(i)]))^\mathcal{I} = (\neg \{i\} \wedge \exists
  R. (\neg{i} \wedge \phi[del_R(i,j,R')]))^\mathcal{I}$.
  \item[{$\exists R. \phi[i \gg j]$:}] Let us assume that there exists
    $(n,n') \in
    R^\mathcal{I'} and n' \in (\phi[i \gg j])^\mathcal{I'}$. Then, either:
\begin{itemize}
\item  $i^\mathcal{I} = j^\mathcal{I}$ and thus
  $(n,n') \in R^\mathcal{I}$
\item or $i^\mathcal{I} \neq j^\mathcal{I}$ and then either:
\begin{itemize}
\item $(n,i^\mathcal{I}) \not\in R^\mathcal{I}$ and thus $(n,n') \in
  R^\mathcal{I}$
\item or $(n,i^\mathcal{I}) \in R^\mathcal{I}$ and $j^\mathcal{I}
  \not\in  (\phi[i \gg j])^\mathcal{I'}$ and thus $n' \neq
  j^\mathcal{I}$ and thus $(n,n') \in R^\mathcal{I}$
\item or $(n,i^\mathcal{I}) \in R^\mathcal{I}$ and $j^\mathcal{I}
  \in  (\phi[i \gg j])^\mathcal{I'}$ and thus $j^\mathcal{I}$ is a witness.
\end{itemize}
\end{itemize}
Thus the rule is correct.
  \item[{$\exists R. \phi[mrg(i,j)]$:}] Let us assume that there exists
    $(n,n') \in R^\mathcal{I'}$ with $n' \in \phi^\mathcal{I'}$ then
    $n \neq j^\mathcal{I}$ and $n' \neq j^\mathcal{I}$. If $n \neq
    i^\mathcal{I}$ and $n' \neq i^\mathcal{I}$, then $(n,n') \in
    R^\mathcal{I}$ thus $n \in (\neg \{j\} \wedge \exists R. (\neg
    \{j\} \wedge \phi[mrg(i,j)]))^\mathcal{I}$. 
If $n' = i^\mathcal{I}$ then
    either $(n,j^\mathcal{I}) \in R^\mathcal{I}$ or $(n,i^\mathcal{I})
    \in R^\mathcal{I}$ and $\{i^\mathcal{I},j^\mathcal{I}\} \cap
    \phi^\mathcal{I'} \neq emptyset$ thus $n \in (\neg \{j\} \wedge
    \exists R. (\{i\} \vee \{j\}) \wedge \exists U. ((\{i\} \vee \{j\})
    \wedge \phi[mrg(i,j)]))^\mathcal{I}$.
If $n = i^\mathcal{I}$ then either $(i^\mathcal{I},n')$ or
$(j^\mathcal{I},n') \in R^\mathcal{I}$ and thus $n \in (\neg \{j\}
\wedge \{i\} \wedge
\exists U. ((\{i\} \vee \{j\}) \wedge \exists R. \phi[mrg(i,j)]))^\mathcal{I}$.
Thus $(\exists R. \phi)[i \gg j]^\mathcal{I} = (\neg \{j\} \wedge (\exists R. (\neg
    \{j\} \wedge \phi[mrg(i,j)]) \vee \exists R. (\neg
    \{j\} \wedge \phi[mrg(i,j)]) \vee \{i\} \wedge
\exists U. ((\{i\} \vee \{j\}) \wedge \exists R. \phi[mrg(i,j)])))^\mathcal{I}$.
  \item[{$\exists R. \phi[cl(i,j,...)]$:}] Let us assume that there
    exists $e' \in E^\mathcal{I'}$ such that $s^\mathcal{I'}(e') =n$,
    $t^\mathcal{I'}(e') = n'$ and $(n,n') \in R^\mathcal{I'}$ then either:
    \begin{itemize}
    \item $e' \in \ein$ and then $n = s^{I}(in(e'))$ and $n' =
      j^\mathcal{I}$, that is there exists $e$ such that
      $(s^\mathcal{I}(e),t^\mathcal{I}(e)) \in R^\mathcal{I}$ and
      $s^\mathcal{I}(e) = n \neq i^\mathcal{I}$ and $t^\mathcal{I}(e)
      = i^\mathcal{I}$. Thus $n \in (\neg \{i\} \wedge \exists R.\{i\} \wedge \exists U.(\{j\} \wedge \phi[cl(i,j,\dots)]))^\mathcal{I}$.
    \item $e' \in \eout$ and then $n = j^\mathcal{I}$  and $n' =
      t^\mathcal{I}(out(e'))$, that is there exists $e$ such that
      $(s^\mathcal{I}(e),t^\mathcal{I}(e)) \in R^\mathcal{I}$ and $s^G(e) =
      i^\mathcal{I}$ and $t^\mathcal{I}(e) = n'$. Thus $n \in (\{j\} \wedge \exists U.(\{i\} \wedge \exists R.(\neg \{i\} \wedge \phi[cl(i,j,\dots)])))^\mathcal{I}$.
    \item $e' \in \elin$ and then $n = i^\mathcal{I}$, $n' =
      j^\mathcal{I}$ and there exists $e$ such that
      $(s^\mathcal{I}(e),t^\mathcal{I}(e)) \in R^\mathcal{I}$ and
      $s^\mathcal{I}(e) = i^\mathcal{I}$ and $t^\mathcal{I}(e) =
      i^\mathcal{I}$. Thus $n \in (\{i\} \wedge \exists R.\{i\} \wedge \exists U.(\{j\} \wedge \phi[cl(i,j,\dots)]))^\mathcal{I}$.
    \item $e' \in \elout$ and then $n = j$, $n' = i$ and there exists
      $e$ such that $(s^\mathcal{I}(e),t^\mathcal{I}(e)) \in
      R^\mathcal{I}$ and $s^\mathcal{I}(e) = i^\mathcal{I}$ and
      $t^\mathcal{I}(e) = i^\mathcal{I}$. Thus $n \in (\{j\} \wedge \exists R.\{i\} \wedge \exists U.(\{i\} \wedge \exists R. \{i\} \wedge \phi[cl(i,j,\dots)]))^\mathcal{I}$.
    \item $e' \in \elol$ and then $n = j$, $n' = j$ and there exists $e$ such that $(s^\mathcal{I}(e),t^\mathcal{I}(e)) \in
      R^\mathcal{I}$ and $s^\mathcal{I}(e) = i^\mathcal{I}$ and
      $t^\mathcal{I}(e) = i^\mathcal{I}$. Thus $n \in (\{j\} \wedge \phi[cl(i,j,\dots)] \wedge \exists U.(\{i\} \wedge \exists R.\{i\}))^\mathcal{I}$.
    \item otherwise, $e' \in E^G$ and thus $n \in (\exists R.(\phi[cl(i,j,\dots)]))^\mathcal{I}$.
    \end{itemize}
\item[{$(<\; n\; R\; \phi)[add_C(i,C_0)]$:}] As the valuation of $R$ is not
  modified by the substitution, $((<\; n\; R\; \phi)[add_C(i,C_0)])^\mathcal{I} = (<\; n\; R\; (\phi[add_C(i,C_0)]))^\mathcal{I}$.
\item[{$(<\; n\; R\; \phi)[del_C(i,C_0)]$:}] As the valuation of $R$ is not
  modified by the substitution, $((<\; n\; R\;
  \phi)[del_C(i,C_0)])^\mathcal{I} = (<\; n\; R\;
  (\phi[del_C(i,C_0)]))^\mathcal{I}$.
\item[{$(<\; n\; R\; \phi)[add_R(i,j,R')]$:}] As the valuation of $R$ is not
  modified by the substitution, $((<\; n\; R\;
  \phi)[add_R(i,j,R')])^\mathcal{I} = (<\; n\; R\;
  (\phi[add_R(i,j,R')]))^\mathcal{I}$.
\item[{$(<\; n\; R\; \phi)[del_R(i,j,R')]$:}] As the valuation of $R$ is not
  modified by the substitution, $((<\; n\; R\;
  \phi)[del_R(i,j,R')])^\mathcal{I} = (<\; n\; R\;
  (\phi[del_R(i,j,R')]))^\mathcal{I}$.
\item[{$(<\; n\; R\; \phi)[add_R(i,j,R)]$:}] Let us assume that $n \in
  (<\; n\; R\; \phi)[add_R(i,j,R)]^\mathcal{I'}$ then either:
\begin{itemize}
\item $n = i^\mathcal{I}$, $j^\mathcal{I} \in
  \phi[add_R(i,j,R)]^\mathcal{I}$ and $(n,j^\mathcal{I}) \not \in
  R^\mathcal{I}$ and thus, $n \in (<\; (n - 1)\; R\;
  \phi[add_R(i,j,R)])^\mathcal{I}$
\item otherwise, the number of neighbors of $n$ is left unchanged and
  thus $n \in (<\; n\; R\;
  \phi[add_R(i,j,R)])^\mathcal{I}$
\end{itemize}
Thus the rule is correct.
\item[{$(<\; n\; R\; \phi)[del_R(i,j,R)]$:}] Let us assume that $n \in
  (<\; n\; R\; \phi)[del_R(i,j,R)]^\mathcal{I'}$ then either:
\begin{itemize}
\item $n = i^\mathcal{I}$, $j^\mathcal{I} \in
  \phi[del_R(i,j,R)]^\mathcal{I}$ and $(n,j^\mathcal{I}) \in
  R^\mathcal{I}$ and thus, $n \in (<\; (n + 1)\; R\;
  \phi[del_R(i,j,R)])^\mathcal{I}$
\item otherwise, the number of neighbors of $n$ is left unchanged and
  thus $n \in (<\; n\; R\;
  \phi[del_R(i,jR)])^\mathcal{I}$
\end{itemize}
Thus the rule is correct.
\item[{$(<\; n\; R\; \phi)[add_N(i)]$:}] As the valuation of $R$ is not
  modified by the substitution, $((<\; n\; R\;
  \phi)[add_N(i)])^\mathcal{I} = (<\; n\; R\;
  (\phi[add_N(i)]))^\mathcal{I}$.
\item[{$(<\; n\; R\; \phi)[del_R(i,j,R')]$:}] As the valuation of $R$
  becomes $R^\mathcal{I} \backslash \{(n,n') | n = i^\mathcal{I}$ or $
  n' = i^\mathcal{I}\}$,$((<\; n\; R\;
  \phi)[del_R(i,j,R')])^\mathcal{I} = (\{i\} \vee (<\; n\; R\;
  (\phi[del_R(i,j,R')] \wedge \neg \{i\}))^\mathcal{I}$.
\item[{$(<\; n\; R\; \phi)[i \gg j)]$:}] Let us consider whether the
  node $m \in (<\; n\; R\; \phi)^\mathcal{I'}$
  gains or loses neighbors in $\phi^\mathcal{I'}$:
\begin{itemize}
\item if $i^\mathcal{I} = j^\mathcal{I}$, the transformation didn't
  change anything.
\item otherwise:
\begin{itemize}
\item if $(m,i^\mathcal{I}) \in R^\mathcal{I}$, $(m,j^\mathcal{I})
  \not\in R^\mathcal{I}$ and $j^\mathcal{I} \not\in
  \phi^\mathcal{I'}$, $m$ lost one and thus had less than $n + 1$,
\item if $(m,i^\mathcal{I}) \in R^\mathcal{I}$ and $(m,j^\mathcal{I})
  \in R^\mathcal{I}$, $m$ lost one and thus had less than $n + 1$,
\item if $(m,i^\mathcal{I}) \not\in R^\mathcal{I}$, $(m,j^\mathcal{I})
  \in R^\mathcal{I}$ and $j^\mathcal{I} \in
  \phi^\mathcal{I'}$, $m$ gained one and thus had less than $n - 1$,
\item otherwise, its number of neighbors in $\phi^\mathcal{I'}$ does
  not change.
\end{itemize}
\end{itemize}
Thus the rule is correct.
\item[{$(<\; n\; R^-\; \phi)[i \gg j)]$:}] Let us consider whether the
  node $m \in (<\; n\; R\; \phi)^\mathcal{I'}$
  gains or loses neighbors in $\phi^\mathcal{I'}$:
\begin{itemize}
\item if $m = i^\mathcal{I} = j^\mathcal{I}$, the transformation didn't
  change anything.
\item if $m = i^\mathcal{I}$, $m$ lost all its neighbors and thus has less than $n$,
\item if $m = j^\mathcal{I}$, $m$ gained all of $i^\mathcal{I}$'s
  neighbors and thus the sum of its neighbors and those of
  $i^\mathcal{I}$ had to be less than $n$,
\item otherwise, its number of neighbors in $\phi^\mathcal{I'}$ does
  not change.
\end{itemize}
Thus the rule is correct.
  \item[{$(<\; n\; R\; C)[mrg(i,j)]$}] Let us consider whether the
  node $m \in (<\; n\; R\; \phi)^\mathcal{I'}$
  gains or loses neighbors in $\phi^\mathcal{I'}$:
\begin{itemize}
\item If $m = j^\mathcal{I}$, it has no remaining neighbor,
\item If $m = i^\mathcal{I}$, it gains all neighbors of
  $j^\mathcal{I}$ and thus the sum of its neighbors and those of
  $j^\mathcal{I}$ had to be less than $n$,
\item otherwise:
\begin{itemize}
\item if $(m,i^\mathcal{I}) \in R^\mathcal{I}$, $(m,j^\mathcal{I})
  \not\in R^\mathcal{I}$, $i^\mathcal{I} \not\in
  \phi^\mathcal{I'}$ and $j^\mathcal{I} \in
  \phi^\mathcal{I'}$, $m$ gained one and thus had less than $n - 1$,
\item if $(m,j^\mathcal{I}) \in R^\mathcal{I}$, $(m,i^\mathcal{I})
  \not\in R^\mathcal{I}$, $j^\mathcal{I} \not\in
  \phi^\mathcal{I'}$ and $i^\mathcal{I} \in
  \phi^\mathcal{I'}$, $m$ gained one and thus had less than $n - 1$,
\item \item if $(m,i^\mathcal{I}) \in R^\mathcal{I}$, $(m,j^\mathcal{I})
  \in R^\mathcal{I}$, $j^\mathcal{I} \in
  \phi^\mathcal{I'}$ and $i^\mathcal{I} \in
  \phi^\mathcal{I'}$, $m$ losed one and thus had less than $n + 1$,
\item otherwise, they stay the same.
\end{itemize}
\end{itemize}
The rule is thus correct.
\item [{$(<\; n\; R\; C)[cl(i,j,\dots)]$}] Let us consider whether the
  node $m \in (<\; n\; R\; \phi)^\mathcal{I'}$
  gains or loses neighbors in $\phi^\mathcal{I'}$:
 \begin{itemize}
 \item If $m = i^\mathcal{I}$, it can only gain one possible neighbor ($j$) if
      $(m,m) \in R^\mathcal{I}$, $j^\mathcal{I} \in \phi^\mathcal{I'}$
      and $R \in \rlin$. In that case, it needs have one less
      neighbor.
 \item If $m = j^\mathcal{I}$, then either:
  \begin{itemize}
  \item $R \not \in \rout$ and thus the only possible neighbors are $i$
  and $j$. Then $m \in (<\; n\; R\; \phi)^\mathcal{I'}$ if and only if
  one of the following is true:
   \begin{itemize}
   \item $n > 2$,
   \item $n > 1$ and $R \not\in \rlout \cap \rll$,
   \item $R \not \in \rlout \cup \rll$,
   \item $n = 1$ and $R \in \rlout \backslash \rll$ and
  $(i^\mathcal{I},i^\mathcal{I}) \not \in R^\mathcal{I}$ or
  $i^\mathcal{I} \not \in \phi^\mathcal{I'}$,
   \item $n = 1$ and $R \in \rll \backslash \rlout$ and
  $(i^\mathcal{I},i^\mathcal{I}) \not \in R^\mathcal{I}$ or
  $j^\mathcal{I} \not \in \phi^\mathcal{I'}$,
   \item $n = 1$ and $R \in \rlout \cap \rll$ and
  $(i^\mathcal{I},i^\mathcal{I}) \not \in R^\mathcal{I}$ or both
  $\{i^\mathcal{I},j^\mathcal{I}\} \cap \phi^\mathcal{I'} =
  \emptyset$,
   \item $n = 2$ and $R \in \rlout \cap \rll$ and either
  $(i^\mathcal{I},i^\mathcal{I}) \not \in R^\mathcal{I}$, $i^\mathcal{I} \not \in \phi^\mathcal{I'}$ or
  $j^\mathcal{I} \not \in \phi^\mathcal{I'}$
   \end{itemize}
  \item $R \in \rout$ and either:
   \begin{itemize}
   \item $R \in \rlout \backslash \rll$. If:
    \begin{itemize}
     \item $i^\mathcal{I}
  \in \phi^\mathcal{I'}$ and $(i^\mathcal{I},i^\mathcal{I}) \in
  R^\mathcal{I}$, $m$ has as many neighbors in $\phi^\mathcal{I'}$
  different from $i^\mathcal{I}$ as $i^\mathcal{I}$ plus
  $i^\mathcal{I}$ and thus $i^\mathcal{I}$ needs haveless than $n-1$,
    \item otherwise, $i^\mathcal{I}$ needs have less than $n$,
    \end{itemize}
   \item $R \in \rll \backslash \rlout$. If:
    \begin{itemize}
     \item $j^\mathcal{I}
  \in \phi^\mathcal{I'}$ and $(i^\mathcal{I},i^\mathcal{I}) \in
  R^\mathcal{I}$, $m$ has as many neighbors in $\phi^\mathcal{I'}$
  different from $i^\mathcal{I}$ as $i^\mathcal{I}$ plus
  itself and thus $i^\mathcal{I}$ needs have less than $n-1$,
    \item otherwise, $i^\mathcal{I}$ needs have less than $n$,
    \end{itemize}
   \item $R \in \rlout \cap \rll$ If:
    \begin{itemize}
    \item $i^\mathcal{I}
  \in \phi^\mathcal{I'}$, $j^\mathcal{I}
  \in \phi^\mathcal{I'}$ and $(i^\mathcal{I},i^\mathcal{I}) \in
  R^\mathcal{I}$, $m$ has as many neighbors in $\phi^\mathcal{I'}$
  different from $i^\mathcal{I}$ as $i^\mathcal{I}$ plus
  $i^\mathcal{I}$ and itself and thus $i^\mathcal{I}$ needs have less
  than $n -2$,
    \item $i^\mathcal{I}
  \not\in \phi^\mathcal{I'}$, $j^\mathcal{I}
  \in \phi^\mathcal{I'}$ and $(i^\mathcal{I},i^\mathcal{I}) \in
  R^\mathcal{I}$, $m$ has as many neighbors in $\phi^\mathcal{I'}$
  different from $i^\mathcal{I}$ as $i^\mathcal{I}$ plus
 itself and thus $i^\mathcal{I}$ needs have less
  than $n -1$,
    \item $i^\mathcal{I}
  \in \phi^\mathcal{I'}$, $j^\mathcal{I}
  \not\in \phi^\mathcal{I'}$ and $(i^\mathcal{I},i^\mathcal{I}) \in
  R^\mathcal{I}$, $m$ has as many neighbors in $\phi^\mathcal{I'}$
  different from $i^\mathcal{I}$ as $i^\mathcal{I}$ plus
 $i^\mathcal{I}$  and thus $i^\mathcal{I}$ needs have less
  than $n -1$,
    \item either $i^\mathcal{I}
  \not\in \phi^\mathcal{I'}$ and $j^\mathcal{I}
  \not\in \phi^\mathcal{I'}$ or $(i^\mathcal{I},i^\mathcal{I}) \not\in
  R^\mathcal{I}$, $m$ has as many neighbors in $\phi^\mathcal{I'}$
  different from $i^\mathcal{I}$ as $i^\mathcal{I}$ and thus $i^\mathcal{I}$ needs have less
  than $n$
    \end{itemize}
   \end{itemize}
  \end{itemize}
 \item If $m \neq i^\mathcal{I}$ and $m \neq j^\mathcal{I}$, it can
   only gain one neighbor ($j^\mathcal{I}$). It only gains it if
   $(m,i^\mathcal{I}) \in R^\mathcal{I}$, $j^\mathcal{I} \in
   \phi^\mathcal{I'}$ and $R \in \rin$. In such a case, $m$ needed
   have less than $n - 1$. Otherwise, it needed have less than $n$.
\end{itemize}
Thus the rule is correct.
\end{description}
\end{proof}

One can observe that the equivalent formula given for $(<\; n\; R\;
C)[mrg(i,j)]$ uses $R^-$, namely in $(<\; k\;R\;(C[mrg(i,j) \wedge
\forall R^-.\neg j))$. If the logic contains counting quantifiers
($\mathcal{Q}$) but not inverse roles ($\mathcal{I}$), we did not
prove that the logic is closed under substitutions. 

\begin{theorem}\label{th:notclosed}
The logics $\mathcal{ALCQUO}$ and $\mathcal{ALCQUOS}elf$ are not
closed under substitutions.
\end{theorem}

In order to prove this theorem, we use the notion of bisimulation \cite{BESDL16}.

\begin{definition}[$\mathcal{ALCQUO}$-Bisimulation]\label{def:bisimulation}
Given a signature (\textbf{C, R, I}) and two interpretations $\mathcal{I}$ and $\mathcal{J}$, a non-empty
binary relation $Z \subseteq (\Delta^\mathcal{I} \times
\Delta^\mathcal{J})$ is an $\mathcal{ALCQUO}$- bisimulation if it satisfies: \\
\begin{tabular}{llll}
\textbf{($\mathcal{ALC}_1$)} & \multicolumn{3}{l}{$d_1 Z d_2 \implies \forall A \in \textbf{C}, (d_1 \in
  A^\mathcal{I} \Leftrightarrow d_2 \in A^\mathcal{J})$}\\
\textbf{($\mathcal{ALC}_2$)} & \multicolumn{3}{l}{$\forall R \in
                               \textbf{R}, ( d_1 Z d_2 \wedge
                               (d_1,e_1) \in R^\mathcal{I} \implies
                               \exists e_2. (d_2,e_2) \in R^\mathcal{J} \wedge e_1 Z e_2)$}\\
\textbf{($\mathcal{ALC}_3$)} & \multicolumn{3}{l}{$\forall R \in
                               \textbf{R}, ( d_1 Z d_2 \wedge
                               (d_2,e_2) \in R^\mathcal{J} \implies
                               \exists e_1. (d_1,e_1) \in  R^\mathcal{I} \wedge e_1 Z e_2)$}\\
\textbf{($\mathcal{ALC}_4$)} & $\forall i \in \textbf{I}, i^\mathcal{I} Z i^\mathcal{J}$
& \textbf{($\mathcal{O}$)} & $\forall i \in \textbf{I}, d_1 Z d_2 \implies (d_1 = i^\mathcal{I}
  \Leftrightarrow d_2 = i^\mathcal{J})$\\
\textbf{($\mathcal{U}_1$)} & $\forall d \in \Delta^\mathcal{I}, \exists d' \in
  \Delta^\mathcal{J}. d Z d'$
&
                                                                       \textbf{($\mathcal{U}_2$)}
                                                                       &$\forall
                                                                       d' \in \Delta^\mathcal{J}, \exists d \in
  \Delta^\mathcal{J}. d Z d'$\\
\textbf{($\mathcal{Q}$)} & \multicolumn{3}{l}{\hspace*{-4mm}$\forall R \in \textbf{R}, (d_1 Z d_2
  \implies$Z is a bijection between the $R$-successors of $d_1$}\\
&  and those of $d_2$)\\
\end{tabular}  
\end{definition}

\begin{figure}
\def\smallscale{0.8}
\begin{center}
\resizebox{7cm}{!}{
\begin{tikzpicture}[->,>=stealth',shorten >=1pt,auto,node distance=2.8cm,
                    semithick,group/.style ={fill=gray!20, node distance=20mm},thickline/.style ={draw, thick, -latex'}]
  \tikzstyle{Cstate}=[circle,fill=red,draw=none,text=white]
  \tikzstyle{NCstate}=[circle,fill=white,draw=none,text=black]
  \tikzstyle{CurrState}=[circle,fill=white,draw=black,text=black]

  \node[NCstate] (I)                     {$d_1:i$};
  \node[NCstate] (J)  [below of=I]{$d_2:j$};
  \node[Cstate] (K)  [right of=I]{$d_3$};
  \node[Cstate] (L)  [right of=J]{$d_4$};
  \node[Cstate] (K')  [below right of=K]{$d'_3$};
  \node[NCstate] (I') [above right of=K']       {$d'_1:i$};
  \node[NCstate] (J')  [below of=I'] {$d'_2:j$};

\begin{pgfonlayer}{background}
\node [group, fit=(I) (J) (K) (L)] (b) {};
\node [group, fit=(I') (J') (K')] (b') {};
\end{pgfonlayer}

  \path (I)  edge node {$R$} (K)
        (J)  edge node {$R$} (L)
  (I') edge node {$R$} (K')
  (J') edge node {$R$} (K');
\path[dashed,latex-latex] (I) edge [bend left]             node {} (I')
              (J) edge [bend right]          node {} (J')
              (K) edge [bend right] node {} (K')
(L) edge [bend left] node {} (K');

\end{tikzpicture}
}
\end{center}
\caption[Bisimilar model and counter-model using $(\geq\; 2\; R\; C)$]{$d_1$ is a model of $(\geq\; 2\; R\; C)[mrg(i,j)]$ and $(\geq\; 2\; R\; C)[i \gg j]$. $d'_1$ is not. Nodes satisfying $C$ are drawn in red.}\label{fig:bisimulation1}
\end{figure}
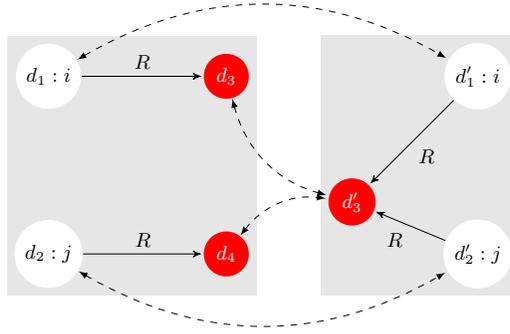

\begin{theorem}\cite{Divroodi2015465}
Let $(\Delta_1, \cdot^{\mathcal{I}_1})$ and
$(\Delta_2, \cdot^{\mathcal{I}_2})$ be two interpretations and $Z$ a $\mathcal{ALCQUO}$-bisimulation relation between
$\mathcal{I}_1$ and $\mathcal{I}_2$. Let $C$ be an $\mathcal{ALCQUO}$ concept, then for all $x_1 \in \Delta_1$ and $x_2 \in \Delta_2$, $x_1 Z x_2 \Rightarrow(x_1 \in C^{\mathcal{I}_1}  \Leftrightarrow x_2 \in C^{\mathcal{I}_2})$.
\end{theorem}

The notion of bisimulation can be extended to $\mathcal{ALCQUOS}elf$
as follows.

\begin{definition}[$\mathcal{ALCQUOS}elf$-Bisimulation]\label{def:sbisimulation}
Given a signature (\textbf{C, R, I}) and two interpretations $\mathcal{I}$ and $\mathcal{J}$, a non-empty
binary relation $\mathcal{Z} \subseteq (\Delta^\mathcal{I} \times
\Delta^\mathcal{J})$ is an $\mathcal{ALCQUOS}elf$- bisimulation if it is an $\mathcal{ALCQUO}$-bisimulation and it satisfies: 
\begin{description}
\item[($\mathcal{S}elf$)] $\forall R \in \textbf{R}, d_1 Z d_2 \implies 
  ((d_1,d_1) \in R^\mathcal{I} \Leftrightarrow (d_2,d_2) \in R^\mathcal{J})$. 
\end{description}
\end{definition}

\begin{proof}[\theoremref{th:notclosed}]
We use the interpretations from \figref{fig:bisimulation1} to show
that some concept with substitution is not in $\mathcal{ALCQUO}$ or
$\mathcal{ALCQUOS}elf$. Let us start by proving that two
interpretations are indeed bisimilar:
\begin{description}
\item[$(\mathcal{ALC}_1)$] :
\begin{itemize}
\item $d_1 Z d'_1 \rightarrow (d_1 \in C^\mathcal{I} \Leftrightarrow
  d'_1 \in C^\mathcal{J}$ \checkmark 
\item $d_2 Z d'_2 \rightarrow (d_2 \in C^\mathcal{I} \Leftrightarrow
  d'_2 \in C^\mathcal{J}$ \checkmark
\item $d_3 Z d'_3 \rightarrow (d_3 \in C^\mathcal{I} \Leftrightarrow
  d'_3 \in C^\mathcal{J}$ \checkmark
\item $d_4 Z d'_3 \rightarrow (d_4 \in C^\mathcal{I} \Leftrightarrow
  d'_3 \in C^\mathcal{J}$ \checkmark
\end{itemize}  
\item[$(\mathcal{ALC}_2)$] :
\begin{itemize}
\item $d_1 Z d'_1 \wedge (d_1,d_3) \in R^\mathcal{I} \rightarrow
  (d'_1,d'_3) \in R^\mathcal{J} \wedge d_3 Z d'_3$ \checkmark
\item $d_2 Z d'_2 \wedge (d_2,d_4) \in R^\mathcal{I} \rightarrow
  (d'_2,d'_4) \in R^\mathcal{J} \wedge d_2 Z d'_2$ \checkmark  
\end{itemize} 
\item[$(\mathcal{ALC}_3)$] :
\begin{itemize}
\item $d_1 Z d'_1 \wedge (d'_1,d'_3) \in R^\mathcal{J} \rightarrow
  (d_1,d_3) \in R^\mathcal{I} \wedge d_3 Z d'_3$ \checkmark
\item $d_2 Z d'_2 \wedge (d'_2,d'_4) \in R^\mathcal{J} \rightarrow
  (d'_2,d'_4) \in R^\mathcal{I} \wedge d_2 Z d'_2$ \checkmark  
\end{itemize}
\item[$(\mathcal{ALC}_4)$] :
\begin{itemize}
\item $i^\mathcal{I} Z i^\mathcal{J}$ \checkmark
\item $j^\mathcal{I} Z j^\mathcal{J}$ \checkmark  
\end{itemize}
\item[$(\mathcal{O})$] :
\begin{itemize}
\item $d_1 Z d'_1 \rightarrow (d_1 = i^\mathcal{I} \Leftrightarrow
  d'_1 = i^\mathcal{J}$ \checkmark
\item $d_2 Z d'_2 \rightarrow (d_2 = i^\mathcal{I} \Leftrightarrow
  d'_2 = i^\mathcal{J}$ \checkmark
\item $d_3 Z d'_3 \rightarrow (d_3 = i^\mathcal{I} \Leftrightarrow
  d'_3 = i^\mathcal{J}$ \checkmark
\item $d_4 Z d'_3 \rightarrow (d_4 = i^\mathcal{I} \Leftrightarrow
  d'_3 = i^\mathcal{J}$ \checkmark
\item $d_1 Z d'_1 \rightarrow (d_1 = j^\mathcal{I} \Leftrightarrow
  d'_1 = j^\mathcal{J}$ \checkmark
\item $d_2 Z d'_2 \rightarrow (d_2 = j^\mathcal{I} \Leftrightarrow
  d'_2 = j^\mathcal{J}$ \checkmark
\item $d_3 Z d'_3 \rightarrow (d_3 = j^\mathcal{I} \Leftrightarrow
  d'_3 = j^\mathcal{J}$ \checkmark
\item $d_4 Z d'_3 \rightarrow (d_4 = j^\mathcal{I} \Leftrightarrow
  d'_3 = j^\mathcal{J}$ \checkmark
\end{itemize}
\item[$(\mathcal{U}_1)$] :
\begin{itemize}
\item $d_1 Z d'_1$ \checkmark
\item $d_2 Z d'_2$ \checkmark
\item $d_3 Z d'_3$ \checkmark
\item $d_4 Z d'_3$ \checkmark
\end{itemize}
\item[$(\mathcal{U}_2)$] :
\begin{itemize}
\item $d_1 Z d'_1$ \checkmark
\item $d_2 Z d'_2$ \checkmark
\item $d_3 Z d'_3$ \checkmark
\end{itemize}
\item[$(\mathcal{Q})$] :
\begin{itemize}
\item $d_1 Z d'_1 \rightarrow Z$ is a one-to-one between $\{d_3\}$ and
  $\{d'_3\}$ \checkmark
\item $d_2 Z d'_2 \rightarrow Z$ is a one-to-one between $\{d_4\}$ and
  $\{d'_3\}$ \checkmark
\item $d_3 Z d'_3 \rightarrow Z$ is a one-to-one between $\emptyset$ and
  $\emptyset$ \checkmark
\item $d_4 Z d'_3 \rightarrow Z$ is a one-to-one between $\emptyset$ and
  $\emptyset$ \checkmark
\end{itemize}
\item[$(\mathcal{S}elf)$] :
\begin{itemize}
\item $d_1 Z d'_1 \rightarrow ((d_1,d_1) \in R^\mathcal{I}
  \Leftrightarrow (d'_1,d'_1) \in R^\mathcal{J})$ \checkmark
\item $d_2 Z d'_2 \rightarrow ((d_2,d_2) \in R^\mathcal{I}
  \Leftrightarrow (d'_2,d'_2) \in R^\mathcal{J})$ \checkmark
\item $d_3 Z d'_3 \rightarrow ((d_3,d_3) \in R^\mathcal{I}
  \Leftrightarrow (d'_3,d'_3) \in R^\mathcal{J})$ \checkmark
\item $d_4 Z d'_4 \rightarrow ((d_4,d_4) \in R^\mathcal{I}
  \Leftrightarrow (d'_4,d'_4) \in R^\mathcal{J})$ \checkmark
\end{itemize}
\end{description}

However, applying the transformation $mrg(i,j)$ yields the
interpretations shown in \figref{fig:bisimulation2} where $d_1$ is a
model of $(\geq\; 2\; R\; C)$ but $d'_1$ is not i.e. $(\geq\; 2\; R\;
C)[mrg(i,j)]$ is not a concept of $\mathcal{ALCQUO}$ or $\mathcal{ALCQUOS}elf$.  
\end{proof}

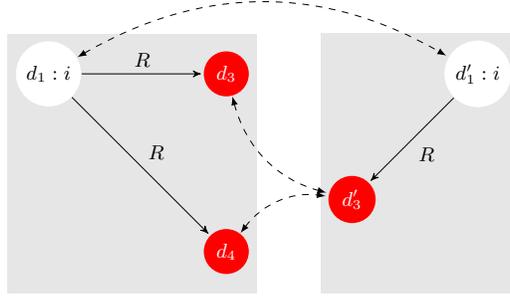
\begin{figure}
\def\smallscale{0.8}
\begin{center}
\resizebox{7cm}{!}{
\begin{tikzpicture}[->,>=stealth',shorten >=1pt,auto,node distance=2.8cm,
                    semithick,group/.style ={fill=gray!20, node distance=20mm},thickline/.style ={draw, thick, -latex'}]
  \tikzstyle{Cstate}=[circle,fill=red,draw=none,text=white]
  \tikzstyle{NCstate}=[circle,fill=white,draw=none,text=black]
  \tikzstyle{CurrState}=[circle,fill=white,draw=black,text=black]

  \node[NCstate] (I)                     {$d_1:i$};
  \node[Cstate] (K)  [right of=I]{$d_3$};
  \node[Cstate] (L)  [right of=J]{$d_4$};
  \node[Cstate] (K')  [below right of=K]{$d'_3$};
  \node[NCstate] (I') [above right of=K']       {$d'_1:i$};

\begin{pgfonlayer}{background}
\node [group, fit=(I) (J) (K) (L)] (b) {};
\node [group, fit=(I') (J') (K')] (b') {};
\end{pgfonlayer}

  \path (I)  edge node {$R$} (K)
        (I)  edge node {$R$} (L)
  (I') edge node {$R$} (K');
\path[dashed,latex-latex] (I) edge [bend left]             node {} (I')
              (K) edge [bend right] node {} (K')
(L) edge [bend left] node {} (K');

\end{tikzpicture}
}
\end{center}
\caption{$d_1$ is a model of $(\geq\; 2\; R\; C)$; $d'_1$ is not. Nodes satisfying $C$ are drawn in red.}\label{fig:bisimulation2}
\end{figure}

\begin{example}
We computed in \exampleref{ex:servernetcorr} the correctness formula for the client-to-proxy connection example. The formula still contained $App(\rho)$ and substitutions, however, that we can now replace with their correct expressions. $App(\rho_0)$ is equivalent to $\exists U. (Client \wedge \exists Request. (Proxy \wedge (<\; N\; C2P^-\; \top)))$ and $App(\rho_1)$ is equivalent to $\exists U. (Client \wedge \exists Request. (Proxy \wedge (\geq\; N\; C2P^-\; \top)))$. After simplification, $wp(\alpha_{\rho_0}, Post)$ is equivalent to $\forall U. (Proxy \wedge Active \Rightarrow ((j \wedge \exists U.(i \wedge \forall C2P.\neg j) \Rightarrow (\leq\; N-1\; C2P^-\; \top)) \wedge (\neg j \vee  \exists U.(i \wedge \exists C2P.j) \Rightarrow (\leq\; N\; C2P^-\; \top)))$ and $wp(\alpha_{\rho_1}, Post)$ is equivalent to $\forall U. ((Proxy \vee (k \wedge \exists U.(j \wedge Proxy))) \wedge (Active \vee k) \Rightarrow ((k \Rightarrow \top) \wedge (\neg k \vee  \exists C2P^-.i) \Rightarrow (\leq\; N\; C2P^-\; \top)))$. 

Proving that the correctness formula is valid amounts to proving that Proxies, including the possible new one $k$, satisfy some conditions. Let us first prove that $Pre \wedge App(\rho_0) \Rightarrow wp(\alpha_{\rho_0},Post)$ is valid:
\begin{itemize}
\item For all Proxies that are not $j$, nothing has changed
\item For $j$, if it had strictly less than $N$ incoming edges labeled with $C2P$, that is if $\rho_0$ was the rule that was applied, it satisfies $(\leq\; N-1\; C2P^-\; \top)$. It thus satisfies $Proxy \wedge Active \Rightarrow (\leq\; N-1\; C2P^-\; \top)$ if there was an edge from $i$ to $j$ labeled with $C2P$, that is if $\exists U.(i \wedge \forall C2P.\neg j)$ is satisfied, and $Proxy \wedge Active \Rightarrow (\leq\; N-1\; C2P^-\; \top)$ if not.
\end{itemize}
$Pre \wedge App(\rho_0) \Rightarrow wp(\alpha_{\rho_0},Post)$ is thus valid. Let us focus now on $Pre \wedge App(\rho_1) \Rightarrow wp(\alpha_{\rho_1},Post)$:
\begin{itemize}
\item For all Proxies that are not $j$ or $k$, nothing has changed
\item $j$, from $Pre$, satifies $(\leq\; N\; C2P^-\; \top)$ and thus $(Proxy \wedge Active \Rightarrow  (\neg k  \Rightarrow (\leq\; N\; C2P^-\; \top))$
\item As for $k$, $k \Rightarrow \top$ is an obvious tautology.
\end{itemize}
As both implications are valid, so is their conjunction and thus the correctness formula is valid. We have successfully proved the correctness of the specification.
\end{example}



\section{Conclusion}\label{sec:conclusion}
We have presented a class of graph rewriting systems, \Grs s, where
the left-hand sides of the considered rules can express additional
application conditions defined as logic formulas and right-hand sides
are sequences of actions.  The considered actions include node mergin
and cloning, node and edge addition and deletion among others.
 We defined computations with these systems by means of
rewrite strategies. There is certainly much work to be done around
such systems with logically decorated left-hand sides. For instance,
the extension to narrowing derivations, which is a matter of future
work, would use an involved unification algorithm taking into account
the underlying logic. We have also presented a sound 
Hoare-like calculus for specifications with pre and post conditions
and shown that the considered correctness problem is still decidable
in most of the logics we used. We also pointed out those logics for
which the rules we gave did not provide a proof of closure under
substitutions and proved that they were not actually closed under
substitutions. Future work include also an implementation of the
proposed verification technique as well as the investigation of more
expressive logics with connections some SMT solvers. 


\bibliographystyle{plain}
 \bibliography{main}

\end{document}